\documentclass[sn-mathphys-num]{sn-jnl}

\usepackage{graphicx}%
\usepackage{multirow}%
\usepackage{amsmath,amssymb,amsfonts}%
\usepackage{amsthm}%
\usepackage{mathrsfs}%
\usepackage[title]{appendix}%
\usepackage{xcolor}%
\usepackage{textcomp}%
\usepackage{manyfoot}%
\usepackage{booktabs}%
\usepackage{algorithm}%
\usepackage{algorithmicx}%
\usepackage{algpseudocode}%
\usepackage{listings}%
\usepackage[defaultlines=3,all]{nowidow}

\usepackage{hyperref}
\hypersetup{
    colorlinks,
    citecolor=blue!80!black,
    filecolor=blue!80!black,
    linkcolor=blue!80!black,
    urlcolor=blue!80!black
}
\usepackage[noabbrev,capitalise,nameinlink]{cleveref}
\crefname{thm}{Theorem}{Theorems}
\crefname{thmC}{Theorem}{Theorems}
\crefname{lem}{Lemma}{Lemmas}
\crefname{drule}{Rule}{Rules}
\usepackage{enumitem}
\usepackage{mathrsfs}
\usepackage{microtype}
\usepackage{xspace}
\usepackage{tcolorbox}
\usepackage{tikz}
\usepackage{todonotes}

\usetikzlibrary{arrows.meta}

\newcommand{\Cc}[0]{\ensuremath{\mathscr{C}}\xspace}

\newcommand{\Gg}[0]{\ensuremath{\mathcal{G}}\xspace}

\renewcommand{\phi}{\varphi}

\newcommand{\Oof}{\mathcal{O}}
\newcommand{\N}{\mathbb{N}}

\renewcommand{\emptyset}{\varnothing}
\renewcommand{\epsilon}{\varepsilon}

\newcommand{\dist}{\ensuremath{\mathrm{dist}}}

\tikzset{
    pics/tile/.style args={#1/#2}{
        code={
            \fill[#1] (0,0) -- (1,0) -- (0.5,0.5) -- cycle;
            \fill[#2] (1,0) -- (1,1) -- (0.5,0.5) -- cycle;
            \fill[#1] (1,1) -- (0,1) -- (0.5,0.5) -- cycle;
            \fill[#2] (0,1) -- (0,0) -- (0.5,0.5) -- cycle;
            \draw (0,0) rectangle (1,1);
        }
    }
}

\newtheorem{theorem}{Theorem}
\newtheorem{lemma}[theorem]{Lemma}
\newtheorem{corollary}[theorem]{Corollary}
\newtheorem{observation}[theorem]{Observation}

\newtheorem{claim}{Claim}

\newtheorem{drule}{Rule}
\newtheorem{subrule}{Rule}[drule]

\begin{document}

\title[Data reduction for DFVS on graphs without long induced cycles]{Data reduction for directed feedback vertex set on\\[2mm] \mbox{}\hspace{10mm}graphs without long induced cycles\\[2mm]
\large{Three rules to rule them all}}


\author[1]{\fnm{Jona} \sur{Dirks}}

\author[1]{\fnm{Enna} \sur{Gerhard}}

\author[1]{\fnm{Mario} \sur{Grobler}}

\author[2]{\fnm{Amer E.} \sur{Mouawad}}

\author*[1]{\fnm{Sebastian} \sur{Siebertz}}\email{siebertz@uni-bremen.de}

\affil[1]{
\orgname{University of Bremen}, 
\country{Germany}}

\affil[2]{
\orgname{American University of Beirut}, \country{Lebanon}}

\abstract{
    We study reduction rules for \textsc{Directed Feedback Vertex Set (DFVS)} on directed graphs without long cycles. 
    A \textsc{DFVS} instance without cycles longer than~$d$ naturally corresponds to an instance of \textsc{$d$-Hitting Set}, however, enumerating all cycles in an $n$-vertex graph and then kernelizing the resulting \textsc{$d$-Hitting Set} instance can be too costly, as already enumerating all cycles can take time~$\Omega(n^d)$.
    To the best of our knowledge, the kernelization of \textsc{DFVS} on graphs without long cycles has not been studied in the literature, except for very restricted cases, e.g., 
    for tournaments, in which all induced cycles are of length three.
    We show that the natural reduction rule to delete all vertices and edges that do not lie on induced cycles cannot be implemented efficiently, that is, it is $W[1]$-hard (with respect to parameter $d$) to decide if a vertex or edge lies on an induced cycle of length at most $d$ even on graphs that become acyclic after the deletion of a single vertex or edge. 
    Based on different reduction rules we then show how to compute a kernel with at most $2^dk^d$ vertices and at most 
    $d^{3d}k^d$ induced cycles of length at most~$d$ (which however, cannot be enumerated efficiently), where~$k$ is the size of a minimum directed feedback vertex set. 
    We then study classes of graphs whose underlying undirected graphs have bounded expansion or are nowhere dense. These are very general classes of sparse graphs, containing e.g.\ classes excluding a minor or a topological minor. 
    We prove that for every class~$\Cc$ with bounded expansion there is a function $f_\Cc(d)$ such that 
    for graphs $G\in \Cc$ without induced cycles of length greater than $d$ 
    we can compute a kernel with $f_\Cc(d)\cdot k$ vertices in time $f_\Cc(d)\cdot n^{\Oof(1)}$. 
    For every nowhere dense class $\Cc$ there is a function $f_\Cc(d,\epsilon)$ such that for graphs $G\in \Cc$ without induced cycles of length greater than $d$ we can compute a kernel with $f_\Cc(d,\epsilon)\cdot k^{1+\epsilon}$ vertices for any $\epsilon>0$ in time $f_\Cc(d,\epsilon)\cdot n^{\Oof(1)}$. 
    %
    %
    The most restricted classes we consider are strongly connected planar graphs without any (induced or non-induced) long cycles. We show that these classes have treewidth $\Oof(d)$ and hence \textsc{DFVS} on planar graphs without cycles of length greater than $d$ can be solved in time $2^{\Oof(d)}\cdot n^{\Oof(1)}$. 
    We finally present a new data reduction rule for general \textsc{DFVS} and prove that the rule together with a few standard rules subsumes all rules applied in the work of Bergougnoux et al.\ to obtain a polynomial kernel for \textsc{DFVS[FVS]}, i.e., \textsc{DFVS} parameterized by the feedback vertex set number of the underlying (undirected) graph.
    We conclude by studying the LP-based approximation of \textsc{DFVS}.
}




\maketitle

\section{Introduction}

A directed feedback vertex set of a directed $n$-vertex graph $G$ is a subset \mbox{$S\subseteq V(G)$} of vertices such that every directed cycle of $G$ intersects with $S$. In 
the \textsc{Directed Feedback Vertex Set (DFVS)} problem, we are given a directed graph $G$ and an integer $k$, and the objective is to determine whether $G$ admits a directed feedback vertex set of size at most $k$. In what follows, unless stated otherwise, 
when we speak of a graph we always mean a directed graph, and when we speak of a cycle we mean a directed cycle. 

\textsc{DFVS} is one of Karp's 21 NP-complete problems~\cite{karp1972reducibility}. Its NP-completeness follows easily by a reduction
from \textsc{Vertex Cover}, which is a special case of \textsc{DFVS} where all induced cycles have length two. The fastest known exact algorithm for \textsc{DFVS}, due to Razgon~\cite{razgon2007computing}, runs in time $1.9977^n\cdot n^{\Oof(1)}$. 
Chen, Liu, Lu, O'Sullivan, and Razgon~\cite{chen2008fixed} proved that the problem is fixed-parameter tractable when parameterized by solution size~$k$; providing an 
algorithm running in time \mbox{$\Oof(k!4^kk^4nm)=2^{\Oof(k\log k)}\cdot nm$}, for graphs with~$n$~vertices and~$m$~edges. The dependence on the input size has been improved to~\mbox{$\Oof(k!4^kk^5(n+m))$} by Lokshtanov, Ramanujan, and Saurabh~\cite{LokshtanovRS16}. It is a major open problem whether the running time can be improved to~$2^{o(k\log k)} \cdot n^{\Oof(1)}$~\cite{LokshtanovRS16}. 
The problem has also been studied under different parameterizations. 
Bonamy et al.~\cite{bonamy2018directed} proved that one can solve the problem in time~$2^{\Oof(t\log t)} \cdot n^{\Oof(1)}$, where~$t$ denotes the treewidth of the underlying undirected graph. 
They also proved that this running time is tight assuming the exponential-time hypothesis~(ETH). 
On planar graphs the running time can be improved to $2^{\Oof(t)} \cdot n^{\Oof(1)}$~\cite{bonamy2018directed}. 

A natural question is whether these results can be extended to directed width measures, e.g., whether the problem is fixed-parameter tractable when parameterized by directed treewidth. 
Unfortunately, this is not the case. \textsc{DFVS} remains NP-complete even on very restricted classes of graphs such as graphs of cycle rank at most four (which in particular have bounded directed treewidth, as cycle rank is an upper bound for directed treewidth), 
as shown by Kreutzer and Ordyniak~\cite{kreutzer2011digraph}, and hence the problem is not even in~XP when parameterized by cycle rank. 

The question whether \textsc{DFVS} parameterized by solution size $k$ admits a polynomial kernel, i.e., an equivalent polynomial-time computable instance of size polynomial in~$k$, 
remains one of the central open questions in the area of kernelization. Bergougnoux et al.~\cite{bergougnoux2021towards} showed 
that the problem admits a kernel of size $\Oof(f^4)$ in general graphs and~$\Oof(f)$ in graphs embeddable on a fixed surface, where $f$ denotes the size of a minimum 
undirected feedback vertex set in the underlying undirected graph. Note that $f$ can be arbitrarily  larger than~$k$. More generally, for an integer $\eta$, a subset $M\subseteq V(G)$ of vertices is called a \emph{treewidth $\eta$-modulator} if $G-M$ has treewidth at most $\eta$. 
Lokshtanov et al.~\cite{lokshtanov2019wannabe} showed that when given a graph $G$, an integer $k$, and a treewidth $\eta$-modulator of size $\ell$, one can compute a kernel with $(k\cdot \ell)^{\Oof(\eta^2)}$ vertices. This result subsumes the result of Bergougnoux et al.~\cite{bergougnoux2021towards}, as the parameter $k+\ell$ is upper bounded by~$\Oof(f)$ and can be arbitrarily smaller than~$f$. 
On the other hand, unless \mbox{$\textsc{NP}\subseteq $ coNP$/$poly}, for $\eta\geq 2$, there cannot exist a polynomial kernel when we parameterize by the size of a treewidth-$\eta$ modulator alone, as even \textsc{Vertex Cover} cannot have a polynomial kernel when parameterized by the size of a treewidth-$2$ modulator~\cite{cygan2014hardness}. Polynomial kernels for \textsc{DFVS} are known for several restricted graph classes, see e.g.~\cite{bang2016algorithms,dom2010fixed,fomin2019subquadratic}.

From the viewpoint of approximation, the best known  algorithms for \textsc{DFVS} are based on integer linear programs (ILP) whose fractional relaxations can be solved efficiently. It was shown by 
Seymour~\cite{seymour1995packing} that the integrality gap for \textsc{DFVS} is at most $\Oof(\log k^*\log\log k^*)$, where~$k^*$ denotes the optimal value of a fractional directed feedback vertex set. Note that the linear programming (LP) formulation of
\textsc{DFVS} may contain an exponential number of constraints. 
Nevertheless, an optimal solution can be computed in this case in polynomial time by using the ellipsoid method. However, using the ellipsoid method in practical applications is often infeasible and much slower than the simplex algorithm, which however needs an explicit formulation of the LP. 
Even et al.~\cite{even1998approximating} completely circumvented this obstacle and provided a related combinatorial polynomial-time algorithm yielding an $\Oof(\log k^*\log\log k^*)$
-approximation. Assuming the Unique Games Conjecture, the problem does not admit a polynomial-time computable constant-factor approximation algorithm~\cite{guruswami2011beating,guruswami2016simple,svensson2012hardness}. 
Lokshtanov et al.~\cite{lokshtanov2021fpt} showed how to compute a $2$-approximation in time $2^{\Oof(k)} \cdot n^{\Oof(1)}$. 

\medskip
This work was initiated after successfully participating in the PACE 2022 programming challenge~\cite{grossmann2022pace}. 
In the scope of a student project at the University of Bremen, we participated in the competition and our solver ranked second in the exact track~\cite{bergenthal2022pace}. 
In this paper we present our theoretical findings, whereas an empirical evaluation of the implemented rules will be presented in future work. 

\medskip
We first study \textsc{DFVS} instances without long cycles. This study is intimately linked to the study of the \textsc{Hitting Set} problem. 
Many of the known data reduction rules for \textsc{DFVS} are special cases of general reduction rules for \textsc{Hitting Set}. 
A hitting set in a set system~$\Gg$ with ground set $V(\Gg)$ and edge set $E(\Gg)$, where each $S\in E(\Gg)$ is a subset of $V(\Gg)$, is a subset $H\subseteq V(\Gg)$ 
such that $H\cap S\neq \emptyset$ for all $S\in E(\Gg)$. Given a graph~$G$, a directed feedback vertex set in $G$ corresponds one-to-one to a hitting set for the set 
system~$\Gg$ where $V(\Gg)=V(G)$ and $E(\Gg)=\{V(C)\mid C$ is a cycle in $G\}$. 
The main difficulty in applying reduction rules designed for \textsc{Hitting Set} is that we first need to efficiently convert an instance of \textsc{DFVS} to an instance of \textsc{Hitting Set}. However, in general, we want to avoid computing~$\Gg$ from $G$, as $|E(\Gg)|$ may be super-polynomial in the size of the vertex set, i.e., super-polynomial in $|V(G)| = |V(\Gg)|$.
One simple reduction rule for \textsc{Hitting Set} is to remove all sets $S\in E(\Gg)$ such that there exists $S'\in E(\Gg)$ with $S'\subseteq S$. 
Instances of \textsc{Hitting Set} that do not contain such pairs of sets are called \emph{vertex induced}. 
The remaining minimal sets in the corresponding \textsc{DFVS} instance are the induced cycles of $G$. 

\medskip
\begin{observation}\label{obs:induced-cycles}
    A set $S\subseteq V(G)$ is a DFVS in a graph $G$ if and only if it hits all induced cycles of $G$. 
\end{observation}
\medskip

Unfortunately, it is NP-complete to detect if a vertex or an edge lies on an induced cycle~\cite{fellows1995complexity} even on planar graphs, implying that it is not easy to exploit this property for \textsc{DFVS} directly. 
Overcoming this obstacle requires designing data reduction rules based on sufficient conditions guaranteeing that a vertex or an edge does not lie on an induced cycle and can therefore be safely removed.

An instance of \textsc{DFVS} without cycles of length greater than $d$ naturally corresponds to an instance of \textsc{$d$-Hitting Set}. 
As shown in~\cite{abu2010kernelization}, 
\textsc{$d$-Hitting Set} admits a kernel with $k+(2d-1)k^{d-1}$ vertices, which can be efficiently computed when the \textsc{$d$-Hitting Set} instance is explicitly given as input. 
This is known to be near optimal, as \mbox{$d$-\textsc{Hitting Set}} does not admit a kernel of size $\Oof(k^{d-\epsilon})$ unless the polynomial hierarchy collapses~\cite{dell2014satisfiability}; note that here size refers to the total size of the instance and not to the number of vertices. 
The question of whether there exists a kernel for $d$-\textsc{Hitting Set} with fewer elements is considered to be one of the most important open problems in kernelization~\cite{bessy2011kernels,dom2010fixed,fomin2023lossy,fomin2019kernelization,you2017approximate}. 
However, even in this restricted case we cannot efficiently generate a (vertex induced) \textsc{$d$-Hitting Set} instance from a \textsc{DFVS} instance, as even testing if a vertex lies on an induced cycle of length at most $d$ is W[1]-hard~\cite{haas2006chordless} when parameterized by $d$. 
We hence have to avoid computing a  \textsc{Hitting Set} instance explicitly but must rather work on the implicit graph representation of a \textsc{DFVS} instance. 
As one of our results, we strengthen the result of~\cite{haas2006chordless} by showing that testing if a vertex or edge lies on an induced cycle of length at most $d$ is W[1]-hard even on graphs that become acyclic after the deletion of a single vertex or edge. 
As an unfortunate consequence, we must assume that the graphs given to our algorithms (together with the parameter $d$) must exclude induced cycles of length greater than $d$, as we cannot verify this condition efficiently. 

To the best of our knowledge, the kernelization of \textsc{DFVS} on graphs without long cycles has not been studied in the literature, except for some very restricted cases, e.g., on tournaments, in which all induced cycles are of length three~\cite{bessy2011kernels,dom2010fixed,fomin2019subquadratic}.

We show that after applying the standard reduction rules,  we can compute in polynomial time a superset $W$ of the vertices that lie on induced cycles of length at most $d$ and which is of size at most $2^dk^d$. 
As it suffices to hit all induced cycles, $G[W]$ is an equivalent instance. 
Up to a factor $k$ and constants depending only on $d$ this matches the best bounds we know for the kernelization of $d$-\textsc{Hitting Set}. 
Our result does not rule out the existence of a smaller kernel, of course. 

In the kernelized instance on $2^dk^d$ vertices we could have $(2^dk^d)^d=2^{d^2}k^{d^2}$ induced cycles. 
Based on the classical sunflower lemma, we prove however, that kernelized instances contain at most $d^{3d}k^d$ induced cycles of length at most $d$, for any fixed $d \geq 2$. 
In light of the major open question whether \textsc{DFVS} admits a polynomial kernel, we pose as a question whether there exists a function $f$ such that for every $d$ we can compute a kernel of size
$f(d)\cdot k^{\Oof(1)}$ in time~$f(d)\cdot n^{\Oof(1)}$ on instances without induced cycles of length greater than $d$. \footnote{
Here and in the following the constants in the $\Oof$-notation do not depend on $d$ or any other parameters, except possibly on the class of graphs we are working on.}

We then turn our attention to restricted graph classes for which we can efficiently test whether a vertex lies on an induced cycle of length at most $d$, e.g.,  by efficient algorithms for first-order model checking~\cite{DreierMS23,grohe2017deciding}.
We study classes of graphs whose underlying undirected graphs have bounded expansion or are nowhere dense. 
These are very general classes of sparse graphs~\cite{nevsetvril2008grad, NesetrilM11a}, including, e.g., all classes that exclude a minor or a topological minor, such as planar graphs. 
We prove that for every class~$\Cc$ with bounded expansion there is a function $f_\Cc(d)$ such that 
for graphs $G\in \Cc$ without induced cycles of length greater than $d$ 
we can compute a kernel with $f_\Cc(d)\cdot k$ vertices in time $f_\Cc(d)\cdot n^{\Oof(1)}$. 
For every nowhere dense class $\Cc$ there is a function $f_\Cc(d,\epsilon)$ such that for graphs $G\in \Cc$ without induced cycles of length greater than $d$ we can compute a kernel with $f_\Cc(d,\epsilon)\cdot k^{1+\epsilon}$ vertices for any $\epsilon>0$ in time $f_\Cc(d,\epsilon)\cdot n^{\Oof(1)}$. 
This answers our above question for very general classes of sparse graph positively. 
Our method is based on the approach in~\cite{DrangeDFKLPPRVS16, EickmeyerGKKPRS17} for the kernelization of the \textsc{Distance-$r$ Dominating Set} problem on bounded expansion and nowhere dense classes. 
%

We conclude our study of restricted graph classes by observing that a strongly connected planar graph without long (induced or non-induced) cycles has bounded treewidth. 
We observe that after the application of the reduction rules, weak components are equal to strong components.
Then, if each strong component has bounded treewidth, the whole graph after application of the rules has bounded treewidth, and the problem is solvable efficiently by the algorithm of Bonamy et al.~\cite{bonamy2018directed}.

We proceed by designing a new data reduction rule that is based on a domination rule for \textsc{Hitting Set}, and which is based on sufficient conditions for vertices and edges to lie on an induced cycle. 
The new rule conveniently generalizes almost all known rules for DFVS. In particular, it subsumes the complicated rules presented by  Bergougnoux et al.~\cite{bergougnoux2021towards} to establish a kernel of size~$\Oof(k\cdot f^3)\in \Oof(f^4)$, where $f$ is the size of a minimum feedback vertex set for the underlying undirected graph. 
In addition to being simpler, our rule does not require the initial computation of a feedback vertex set for the underlying undirected graph.

Finally, we study the LP-based approximation of \textsc{DFVS}. As previously mentioned, we can formulate an ILP that is equivalent to \textsc{DFVS}. 
In the natural formulation, which we call the \emph{cycles ILP}, we introduce a binary variable~$d_v$ for every $v \in V(G)$ where $d_v = 1$ means that $v$ is part of a solution. 
The goal is to minimize the number of variables set to $1$, given that all induced cycles are hit. Note that this formulation can have an exponential number of constraints. 
Instead, we study an equivalent ILP of polynomial size, the \emph{order ILP}, which is based on the fact that a directed graph is acyclic if and only if there is a topological order on its vertex set. We prove that the optimal solution to the LP-relaxation of the order ILP is at most $3$ times smaller than the optimal solution to the LP-relaxation of the cycles~ILP.
This makes LP-based approximation approaches directly accessible to LP solvers (such as CPLEX) and avoids the specialized combinatorial algorithm of Even et al.~\cite{even1998approximating}, and the ellipsoid method. 

\section{Preliminaries}

A \emph{graph $G$} consists of a (non-empty) \emph{vertex set $V(G)$} and \emph{edge set} $E(G) \subseteq V(G) \times V(G)$. For vertices $u,v\in V(G)$ we write $uv$ for the edge directed from~$u$ to $v$. An edge $vv$ is called a \emph{loop}. We denote the \emph{in-} and \emph{out-neighborhood} of $v \in V(G)$  by $N^-_G(v) = \{u \mid uv \in E(G)\}$ and $N^+_G(v) = \{u \mid vu \in E(G)\}$, respectively. 
The \emph{neighborhood} of~$v$ is denoted by $N_G(v) = N^-_G(v) \cup N^+_G(v)$. 
A \emph{cycle} $C$ in a graph $G$ is a subgraph with vertices $v_1,\ldots, v_\ell$ and edges $v_iv_{i+1}$ for $1\leq i<\ell$ and $v_\ell v_1$. 
We denote by~$\ell$ the \emph{length} of~$C$, that is, the number of edges (or vertices) of $C$.
A \emph{path} $P$ in a graph $G$ is a subgraph with vertices $v_1,\ldots, v_{\ell+1}$ and edges $v_iv_{i+1}$ for $1\leq i\leq \ell$. 
$P$ is a \emph{$u$-$v$-path} if its first vertex $v_1$ is $u$ and its last vertex $v_{\ell+1}$ is $v$ and we say that it connects $u$ with~$v$ or that $u$ reaches $v$ by $P$. 
We call~$\ell$ the \emph{length} of the path, that is, the number of edges of $P$. 
Note that every vertex reaches itself by a path of length $0$. 
If there exists a $u$-$v$-path~$P$ and a $v$-$w$-path~$Q$ in $G$, then there also exists a $u$-$w$-path in $G$ that uses only vertices of $P$ and $Q$. 
Slightly abusing notation, because this path may not be unique, we denote such a path by $PQ$, and will be more precise whenever it matters.
For a set $S \subseteq V(G)$ and $d\geq 1$, we write $N^{d+}_G[S]$ to denote the \emph{$d$-out-neighborhood} of $S$ in~$G$ and $N^{d-}_G[S]$ to denote the \emph{$d$-in-neighborhood} of $S$ in $G$. That is, $N^{d+}_G[S]$ contains all vertices of $G$ that are reachable from some vertex in $S$ via a path of length at most $d$, and $N^{d-}_G[S]$ contains all vertices of $G$ that can reach some vertex in $S$ via a path of length at most~$d$. 
Note that since paths of length zero are allowed, we have $S \subseteq N^{d+}_G[S]\cap N^{d-}_G[S]$. We write $N^{d+}_G[v]$ and $N^{d-}_G[v]$ whenever $S = \{v\}$.

For a vertex subset $U \subseteq V(G)$, we denote by~$G[U]$ the graph \emph{induced by $U$}, that is the graph obtained from $G$ where we only keep the vertices in $U$ and the edges with both endpoints in $U$. 
We write $G-U$ for the graph $G[V(G)\setminus U]$ and for a singleton vertex set $\{v\}$ we write $G-v$ instead of $G-\{v\}$. For an edge $uv$ we write $G+uv$ and $G-uv$ for the graph obtained by adding or removing the edge $uv$, respectively. 
A cycle~$C$ is an induced cycle in $G$ if the graph~$G[V(C)]$ is a cycle and a path~$P$ is an induced path if the graph $G[V(P)]$ is a path.
An \emph{almost induced $u$-$v$-path} is a subgraph that induces a path in $G-vu$. 
Slightly abusing notation, when $u,v$ are distinct vertices on an induced cycle $C$, we will say that $C$ decomposes into an induced $u$-$v$-path and an induced $v$-$u$-path, even though this is not true if $vu\in E(G)$ (and hence $vu \in E(G[V(C)]$), in which case the $u$-$v$-path is only almost induced.

We call a set $S \subseteq V(G)$ a \emph{directed feedback vertex set}, dfvs for short, if $G-S$ does not contain any (directed) cycles. An input of the \textsc{Directed Feedback Vertex Set} problem (DFVS) consists of a graph~$G$ and a positive integer $k$. The goal is to determine whether $G$ admits a dfvs of size at most~$k$. We will constantly make use of the following simple lemma. 

\medskip
\begin{lemma}\label{lem:approx-dfvs}
    Let $G$ be a graph without induced cycles of length greater than $d$ and let~$k$ be a positive integer. Then, we can compute in polynomial time either a dfvs of size at most $dk$ or decide that there does not exist a dfvs of size at most~$k$ in~$G$. 
\end{lemma}

\begin{proof}
By greedily packing induced (pairwise) vertex-disjoint cycles we can find at most $k$ cycles (each consisting of at most $d$ vertices), otherwise we can conclude that no solution of size at most $k$ exists. Then, all other (non-packed) induced cycles intersect one of the packed cycles since $G$ does not contain induced cycles of length greater than~$d$. Hence, the vertex set of the packed cycles forms a dfvs of size at most $dk$.  
 
Observe that every cycle contains an induced cycle. Hence, even though we cannot decide efficiently whether a vertex lies on an induced cycle we can efficiently pack induced pairwise vertex-disjoint cycles as needed. 
\end{proof}

A \emph{hypergraph} (also called a \emph{set system}) $\Gg$ consists of a (non-empty) \emph{vertex set}~$V(\Gg)$ and \emph{hyperedge set} 
$E(\Gg) = \{S_1, \dots, S_m\}$, where $S_i \subseteq V(\Gg)$ for all $i \leq m$. 
For a vertex subset $U \subseteq V(\Gg)$, we denote by $\Gg[U]$ the hypergraph~$\Gg$ \emph{induced by~$U$}, that is, the hypergraph with vertex set $U$ and 
hyperedge set $\{S_i \in E(G) \mid S_i \subseteq U\}$. Note that we only keep those hyperedges that are fully contained in~$U$. 
A \emph{hitting set} of a hypergraph~$\Gg$ is a set $H \subseteq V(\Gg)$ such 
that $H \cap S_i \neq \emptyset$ for all $i \leq m$. In other words, $H$ contains at least one vertex from every hyperedge.
The input of a \textsc{Hitting Set} instance consists of a hypergraph~$\Gg$ and a positive integer~$k$, where $E(\Gg)$ explicitly enumerates all sets. The goal is to determine whether~$\Gg$ admits a hitting set of size at most $k$. 
Given a graph~$G$, a directed feedback vertex set in~$G$ corresponds one-to-one to a hitting set for the set 
system~$\Gg$ with $V(\Gg)=V(G)$ and $E(\Gg)=\{V(C)\mid C$ is a cycle in $G\}$. 
In the following, with any graph $G$ we associate the corresponding hypergraph~$\Gg$. We call~$\Gg$ \emph{vertex induced} if there are no two sets $S, S' \in E(\Gg)$ with $S' \subseteq S$.

\section{Reduction rules}

In this section we present our new reduction rule. 
Together with two other standard reduction rules the rule subsumes most of the known rules for DFVS. 
We also list and attribute special cases that are known from the literature and more efficiently computable. 
We use the convention of numbering special cases of a Rule x as Rules x.1, x.2, etc. 

Our new rule is based on a folklore rule for \textsc{Hitting Set}, which in the literature is often attributed to~\cite{weihe1998covering}. 
If there are two vertices $u,v \in V(\Gg)$ such that $u$ appears in every hyperedge in which $v$ appears, then remove $v$ and remove $v$ from all remaining hyperedges. 
We say that $u$ \emph{dominates} $v$. 

This rule translates to DFVS as follows.
By \Cref{obs:induced-cycles}, if there is a vertex~$v$ such that all induced cycles containing $v$ also contain a vertex $u$, then we may remove~$v$ from all cycles. We cannot simply delete $v$ though, but have to take care that all induced cycles originally containing $v$ are still present (just without $v$). 
This corresponds to the following operation in $G$. 
For $v\in V(G)$, we write $G\ominus v$ for the graph obtained by connecting all in-neighbors of~$v$ with all out-neighbors of~$v$ (only adding missing edges) and then removing $v$, or simply removing $v$ if it has no in- or out-neighbors. 
We say that $G\ominus v$ is obtained from $G$ by \emph{shortcutting}~$v$. 



However, we cannot detect efficiently if a vertex $u$ dominates a vertex $v$ as described above. Hence, we work with the following approximation leading to our new rule. 
%

\begin{tcolorbox}
A cycle $C$ is \emph{induced around a vertex $v$} if it contains $v$ and is an induced cycle of length at most $4$, or when all of the following conditions are true:\\[-3mm]
\begin{itemize}
    \item $C=v_1,\ldots, v_\ell$, $\ell \geq 5$, 
    \item $v_3=v$, 
    \item $v_1,v_2,v_3,v_4,v_5$ (almost) induces a path in $G$, and
    \item there is no edge $v_2v_i$ for $4\leq i\leq \ell$. 
\end{itemize}
\end{tcolorbox}

\begin{figure}
    \centering
    \begin{tikzpicture}
    \foreach \i in {1, ..., 5} {
      \fill (\i, 0) circle (2pt) node (\i) {};
    }
    \path[-{Latex}]
      (1) edge (2)
      (2) edge (3)
      (3) edge (4)
      (4) edge (5)
    ;
    \node[below=0cm of 3] {$v$};

    \path[dotted, red, -{Stealth}]
      (1) edge[bend left]  (4)
      (3) edge[bend right] (2)
      (2) edge             (4, -1)
    ;

    \draw[dashed, -{Latex}] plot [smooth] coordinates {(5) (4, -1) (2, -1) (1)};
    \end{tikzpicture}
    \caption{Illustration of a cycle induced around a vertex $v$. Dotted red edges are forbidden.}
    \label{fig:induced-around}
\end{figure}

\medskip
\begin{drule}\label{rule:shortcut-new}
If there are distinct vertices $u, v \in V(G)$ such that every cycle that is induced around $v$ also contains $u$, then shortcut $v$ in $G$.
\end{drule}


\medskip
\begin{lemma}\label{lem:shortcut}
\cref{rule:shortcut-new} is safe and can be implemented in time \mbox{$\Oof(n^6(n+m))$}.
\end{lemma}
\begin{proof}
We begin with the proof of safeness.
Let $S$ be a dfvs of $G$. 
We may assume, without loss of generality, that $S$ does not contain $v$. 
This is because $S$ is a dfvs if and only if it hits all induced cycles by \cref{obs:induced-cycles} and all induced cycles containing $v$ in particular contain $u$, hence, we could replace $v$ by $u$ in $S$. That is, $(S \setminus \{v\}) \cup \{u\}$ is also a dfvs of $G$. 

Let $C'$ be an induced cycle in $G\ominus v$. If $C'$ is not affected by the shortcutting of~$v$, that is, if $C'$ does not contain an in-neighbor $x$ and an out-neighbor $y$ of~$v$, or it does contain such vertices $x$ and $y$ but the edge $xy$ was already present in~$E(G)$, then $C'$ is also an induced cycle of $G$, hence is hit by $S$ in $G$ and $G\ominus v$. This implies that $S$ is also a dfvs in $G\ominus v$. 

Assume $C'$ contains an in-neighbor $x$ and an out-neighbor $y$ of $v$ and the edge $xy$ is a newly introduced shortcut edge $xy$ in $G\ominus v$. 
First, observe that this edge is the only shortcut edge on $C'$. Assume otherwise that $C'$ contains at least two shortcut edges~$xy$ and $x'y'$, say, $C'=xyPx'y'Qx$. By construction, $G\ominus v$ also contains the edges $xy'$ and $x'y$. 
Then $xy'Qx$ is a cycle whose vertex set is a subset of the vertex set of $C'$, contradicting the fact that $C'$ is an induced cycle. 
Thus $V(C)=V(C')\cup \{v\}$ induces a unique cycle in $G$. As~$C$ is induced in~$G$ and $S$ does not contain $v$ by assumption, $C$ and (consequently) $C'$ are hit by~$S$. 
Hence, $S$ is also a dfvs in $G\ominus v$. 

Conversely, let $S'$ be a dfvs of $G\ominus v$. 
If $C$ is an induced cycle of $G$ then $V(C)\setminus\{v\}$ is the vertex set of a cycle in $G\ominus v$, and is therefore hit by $S'$. Note that $V(C)\setminus\{v\}$ is the vertex set of a cycle in $G\ominus v$ regardless of whether $C$ contains $v$, or an in-neighbor and an out-neighbor of $v$, or $C$ is unaffected by the shortcutting of $v$. Hence, $S'$ intersects with $V(C)\setminus\{v\}$ in $G\ominus v$  and is also a dfvs in $G$.

We show how the rule can be implemented efficiently. 
For a pair $u,v$ of vertices we guess the vertices $v_1,v_2,v_4,v_5$ and delete all out-neighbors of $v_2$ except for $v_3$. 
Observe that the condition that every cycle containing $v$ also contains $u$ is equivalent to the condition that $v$ does not lie on a cycle of $G-u$. 
Hence, we also delete the vertex $u$. 
We test wether we can reach $v_1$ from $v_5$ the resulting graph. 
If it does not, then we shortcut~$v$. 
Testing this for all pairs $u,v$ takes time $\Oof(n^2\cdot n^4(n+m))=\Oof(n^6(n+m))$.
\end{proof}

We mention several special cases of \cref{rule:shortcut-new} that can be implemented more efficiently.
The first special case involves all cycles (and is hence more restrictive). The rule can be implemented in time $\Oof(n^2\cdot (n+m))$. 
\medskip
\begin{subrule}\label{rule:shortcut}
If there are distinct vertices $u, v \in V(G)$ such that $u$ appears in every cycle in which $v$ appears, then shortcut $v$ in $G$.
\end{subrule}
\medskip


The second special case is presented as Rule 5 in~\cite{fleischer2009experimental}. The rule can be implemented in time $\Oof(n(n+m))$. 

\medskip
\begin{subrule}
\label{rule:2-disjoint-cycles}
  If $v\in V(G)$ does not lie on two cycles that are vertex-disjoint except for~$v$, then shortcut $v$.   
\end{subrule}
\medskip

To see why \cref{rule:2-disjoint-cycles} is a special case of \cref{rule:shortcut-new} (in fact, of \cref{rule:shortcut}) apply Menger's theorem to find a vertex $u\neq v$ that hits all cycles on which $v$ lies. 
Here, and in all future applications of Menger's theorem, to find a set of vertices that intersects with all cycles containing a vertex~$v$ we construct the following graph $G'$.
We make a ``copy''~$v'$ of~$v$, then delete all outgoing edges of~$v$ and make them outgoing edges of~$v'$ instead, i.e., the out-neighbors of $v$ become out-neighbors of $v'$ instead. We complete the construction by making $v'$ the unique out-neighbor of $v$. 
Then, the cycles containing~$v$ in $G$ correspond one-to-one to the $v'$-$v$-paths in $G'$. 
By Menger's theorem, the size of the minimum vertex cut for~$v'$ and~$v$ (the minimum size of a set of vertices, distinct from~$v'$ and $v$, whose removal disconnects~$v'$ and~$v$) is equal to the maximum number of pairwise internally vertex-disjoint paths from~$v'$ to $v$.
By finding a minimum vertex cut (using a flow algorithm), we find a set of vertices intersecting all cycles in~$G$ that contain $v$. In particular, the vertices of a minimum cut intersect all cycles in~$G$ that are pairwise vertex-disjoint except for $v$. 
Coming back to the rule, a vertex $u\neq v$ that hits all cycles on which $v$ lies must exist if there exist no two cycles that are vertex-disjoint except for $v$. In other words, there exist no two internally vertex-disjoint paths from $v'$ to $v$ and removing $u$ disconnects $v'$ and $v$. This vertex $u$ in fact dominates $v$.

The final two special cases are, in fact, special cases of the aforementioned rule. The following rule is presented, e.g., as 
Rule 1 in~\cite{bergougnoux2021towards} and Rule 3 in~\cite{fleischer2009experimental}. The rule can be implemented in time $\Oof(n+m)$. 

\medskip
\begin{subrule}
\label{rule:isolated}
If $v \in V(G)$ has no in- or no out-neighbor, then $v$ can be removed from the graph $G$. 
\end{subrule}
\medskip

The final special case occurs when $u$ is the only in-neighbor of $v$ or when~$v$ is the only out-neighbor of $u$. The corresponding rule is presented as Rule~3 in~\cite{bergougnoux2021towards} and Rule~4 in~\cite{fleischer2009experimental}. The rule can be implemented in time $\Oof(n+m)$.

\medskip
\begin{subrule}
\label{rule:Y-shape}
If a vertex $v \in V(G)$ has only one in-neighbor or one out-neighbor then shortcut $v$ in $G$.   
\end{subrule}
\medskip


Note that we do not require to exhaustively apply any of the special rules in order to ensure correctness of a more general rule. 
However, this may be beneficial for practical implementations, as the special cases can be carried out more efficiently than the more general rules. 

The next rule is a standard rule, presented e.g., as Rule~1 in~\cite{fleischer2009experimental}. It can be implemented in time $\Oof(m)$. 

\medskip
\begin{drule}\label{rule:loops}
    If $v \in V(G)$ lies on a loop, then add $v$ to the solution, remove~$v$ from~$G$, and decrease the parameter by one.
\end{drule}
\medskip





The following sunflower-like rule was presented as Rule 3 in \cite{bergougnoux2021towards} and the special case of $u=v$ as Rule 6 in \cite{fleischer2009experimental}. 
The rule can be implemented in time $\Oof(kn^2(n+m)$ by iterating through all pairs $u,v$ of vertices and carry out at most $k+1$ steps of the Edmonds-Karp algorithm, each step running in time $\Oof(n+m)$.

\medskip
\begin{drule}\label{rule:disjoint-cycles}
Let $u, v \in V(G)$. If $uv \not\in E(G)$ and there exist more than $k$ internally vertex-disjoint $u$-$v$-paths, then insert the edge $uv$.
\end{drule}
\medskip

Observe that every rule either removes a vertex or adds an edge, hence \cref{rule:shortcut-new} and \cref{rule:loops} can be carried out exhaustively in time at most $\Oof(n^7(n+m))$. 
\cref{rule:disjoint-cycles} can be applied at most $n^2$ times, hence, can be carried out exhaustively in time at most $\Oof(n^2\cdot k\cdot n^2(n+m))\subseteq \Oof(n^5(n+m))$. 
In total the running is bounded by $\Oof(n^7(n+m))$. 

We provide an analysis of the kernel size with respect to the parameter $f$ (size of a smallest feedback vertex set of the underlying undirected graph) in \cref{sec:new-red-rules}. 

Also observe that the introduction of edges cannot create longer induced cycles, so that the rules do not interfere with our assumption in the next section that the input graphs do not contain induced cycles of a certain length. 

We note one more observation. 

\medskip
\begin{lemma}\label{lem:strong-components}
After exhaustive application of \cref{rule:shortcut-new} every vertex $v$ lies on a cycle that is induced around $v$. 
In particular, every weak component is strongly connected.
\end{lemma}
\begin{proof}
The first statement is clear. 
For the last statement let $u_1\ldots u_t$ be a path in the undirected underlying graph of $G$. 
Each of the edges $u_iu_{i+1}$ lies on a directed cycle~$C_i$ of $G$. 
Then in $C_1\ldots C_{t-1}$ we find a directed $u_1$-$u_t$-path; we first find a directed walk by appropriately gluing parts of the cycles and then find a path contained in the walk.
\end{proof}
\section{DFVS in graphs without long induced cycles}

We begin our study of \textsc{DFVS} in graphs without induced cycles of length greater than~$d$. 
We remark that $d$ must be given with the input and that our algorithms are based on the assumption that the input graph does not contain induced cycles of length greater than $d$. 
In the following we fix some number $d$ and assume that all graphs have no induced cycles of length greater than $d$. 

As mentioned in the introduction we would like to delete all vertices that do not lie on induced cycles. 
Unfortunately however, it is NP-complete to determine if a vertex lies on an induced cycle~\cite{fellows1995complexity}. 
In fact, this is even W[1]-hard when parameterized by~$d$~\cite{haas2006chordless}. 

One necessary condition for a vertex $w$ to lie on an induced cycle is as follows. 
By \cref{lem:approx-dfvs} we can approximate a small dfvs $S$. 
If $w$ does not lie on an (almost) induced path of length at most $d$ between two vertices $x,y\in S$ (making a copy of $x$ when dealing with the case $x=y$), then $w$ cannot lie on an induced cycle in $G$. 
Hence, if this is not the case, we can delete $w$.
Since $G-S$ is acyclic, one could hope that we can test this property in time $f(d)\cdot n^{\Oof(1)}$ for some function $f$, however, as we show (in~\cref{lem-whard-induced}) even this is not possible.
The \textsc{Directed Chordless $(s,v,t)$-Path} problem asks, given a graph $G$, vertices $s,v,t$, and integer $d$, whether there exists an induced $s$-$t$-path in $G$ of length at most~$d$ containing~$v$.  
The W[1]-hardness of the problem on general (directed and undirected) graphs was proved in~\cite{haas2006chordless}. 
We show hardness on directed acyclic graphs via a reduction from \textsc{Grid Tiling}. 
To not disturb the flow of the paper we postpone the proof to~\cref{sec:hardness}. 

We hence have to come up with new reduction rules that are efficiently implementable. 
We start with a high-level description of our strategy as well as the obstacles that we need to overcome. 
As outlined above, given a reduced graph $G$ (on which none of the reduction rules is applicable), we first compute a dfvs $S$ of size at most~$dk$ as guaranteed by~\cref{lem:approx-dfvs}. 

\medskip
\begin{drule}\label{rule:short-cycle}
    If a vertex $v \in V(G)$ does not lie on a path of length at most $d$ between any two vertices $x,y \in S$, then delete $v$.
\end{drule}

\medskip
\begin{lemma}
    \cref{rule:short-cycle} is safe and can be implemented in time $\Oof(n+m)$.
\end{lemma}
\begin{proof}
    We describe an algorithm implementing \Cref{rule:short-cycle} in time $\Oof(n+m)$.
    Given a reduced graph $G$ and a dfvs $S$ of size at most $dk$, we first compute a graph $G'$
    as follows. 
    We start with the graph $G - S$ and add two fresh vertices $s, t$ with additional edges $\{sv \mid \text{there is $x \in S$ with $xv \in E(G)$}\}$ and $\{wt \mid \text{there is $y \in S$ with $wy \in E(G)$}\}$. Obviously, $G'$ can be constructed in time $\Oof(n+m)$. 
    Then, for every $v \in V(G') \setminus \{s,t\}$, we compute $\dist(s,v)$ and $\dist(v, t)$. 
    As $G'$ is acyclic this can be done by topological sort in time $\Oof(n+m)$ (started once from $s$ and once from $t$ traversing the edges in the opposite direction). 
    If $\dist(s,v) + \dist(v,t) > d$, we conclude that $v$ does not lie on a path of length at most $d$ between any two vertices in $S$ and hence remove it.
    The correctness of this algorithm follows from the observation that $G'$ is acyclic and the following claim.
    \begin{claim}
     Let $G'$ be a DAG with distinct vertices $s,v,t \in V(G')$. Let $V_{sv}, V_{vt} \subseteq V(G')$ denote the sets of all vertices of all $s$-$v$-paths resp.\ all $v$-$t$-paths. Then $V_{sv} \cap V_{vt} \subseteq \{v\}$. 
    \end{claim}
    \begin{proof}
        Towards a contradiction assume there is a vertex $x \neq v \in V_{sv} \cap V_{vt}$. By definition there are paths $v_1 \dots v_m$ and $v'_1 \dots v'_n$ with $v_1 = s$, $v'_n = t$, $v_m = v'_1 = v$ such that there are indices $i < m$ and $j < n$ with $v_i = v'_j = x$ in $G'$.
        Then $G[\{v_i v_{i+1} \dots v v'_2 \dots v'_j\}]$ contains a cycle, contradicting that $G'$ is acyclic.
    \end{proof}
    As an immediate consequence of the claim we observe that $\dist(s,v) + \dist(v,t) > d$ if and only if $v$ does not lie on a path of length at most $d$ between $s$ and $t$. As $s$ inherits all outgoing edges from the vertices in $S$, and $t$ inherits all incoming edges to vertices in $S$, the correctness of the algorithm follows.

    Finally, we show that \cref{rule:short-cycle} is safe by showing that every vertex that does not lie on a path of length at most $d$ between any two vertices $x,y \in S$ (possibly $x=y$) does not lie on any induced cycle in $G$. 
    In the contrapositive, if a vertex $v$ lies on an induced cycle $C$, then there exist $x,y\in S$ with $\dist(x,v) + \dist(v, y) \leq d$. 
    As~$S$ is a dfvs there are $x,y\in S$ with $x,y\in V(C)$. 
    Choose such $x,y\in S$ such that the subpath on $C$ from~$x$ to $v$ and the subpath from $v$ to $y$ does not contain another vertex of $S$.  
    As $C$ is an induced cycle it has at most $d$ vertices (by assumption~$G$ does not contain induced cycles of length greater than $d$). 
    These paths witness that $\dist(x,v)+\dist(v,y)\leq d$. 
\end{proof}

We now partition the vertex set of $G$ 
into $S$ and $R = V(G) \setminus S$. 
We have $|S| \leq dk$, $G[R]$ is acyclic, and by the above preprocessing every vertex in $R$ is at distance at most~$d$ to and from some vertex in $S$. 
We now would like to check for each $w\in R$ whether there exists an induced path of length at most~$d$ from a vertex in $S$ to another (or the same) vertex of $S$. 
However, due to~\cref{lem-whard-induced} we cannot check this property efficiently. 
Our solution consists of adopting a ``relaxed approach''. 
That is, let $I^d_u \subseteq V(G)$ denote the set of all vertices that belong to some induced cycle of length at most~$d$ that also includes a vertex $u \in S$. 
We shall compute, for each vertex $u \in S$, a set $W^d_u \supseteq I^d_u$. 
In other words, we compute a superset of $I^d_u$, which we call $W^d_u$, of the vertices that share an induced cycle of length at most~$d$ with $u$. 
We call $W^d_u$ the set of \emph{$d$-weakly relevant vertices for~$u$}. 
Most crucially, we show that each set $W^d_u$ can be computed efficiently and will be of bounded size. 
We define $W^d_S = \bigcup_{u \in S}{W^d_u}$ and we call $W^d_S$ the set of \emph{$d$-weakly relevant vertices for $S$}. 
It is not hard to see that $G[S \cup W^d_S]$ is indeed an equivalent instance (to~$G$) as it includes all vertices that participate in induced cycles of length at most~$d$. 

We describe the construction of $W^d_u$ for a single vertex. That is, we fix a non-reducible directed graph $G$, an integer $k \geq 2$, a constant $d \geq 2$, a dfvs $S$ of size at most~$dk$, and a vertex $u \in S$. We first construct a graph $H^d_u$ as follows:\\[-2mm]

\begin{itemize}
    \item We begin by setting $H^d_u = G[N^{d+}_G[u]]$.\\[-3mm]
    \item Then, we add a new vertex $v$ to $H^d_u$ and make all the in-neighbors of $u$ become in-neighbors of $v$ instead, i.e, $u$ will only have out-neighbors and $v$ will only have in-neighbors. \\[-3mm]
    \item Next, we delete all vertices in $H^d_u$ that do not belong to some directed path from $u$ to $v$ of length at most~$d$. \\[-3mm]
\end{itemize}

Note that $H^d_u$ can be computed in polynomial time. Moreover, there exists an induced cycle of length at most~$d$ containing $u$ in $G$ if and only if there exists an induced~$u$ to~$v$ path of length at most~$d$ in $H^d_{u}$. 
By a slight abuse of notation, we also denote the graph~$H^d_u$ by~$H^d_{u,v}$ to emphasize the source and sink vertices. 
We call a directed graph \emph{$k$-nice} whenever any two vertices $x,z$ are connected by the directed edge~$xz$ or by a set of at most $k$ pairwise internally vertex-disjoint $x$-$z$-paths. In particular, $xz$ is an edge or there exists a set $Y$ (disjoint from $\{x,z\}$) of at most $k$ vertices that hits every directed path from~$x$ to $z$. 
Observe that $H^d_u$ is indeed $k$-nice (since~\cref{rule:disjoint-cycles} has been exhaustively applied on~$G$). 

Given a $k$-nice graph $H^d_u$, two vertices $x,z \in V(H^d_u)$, and $2 \leq d' < d$, we let~$H^{d'}_{x,z}$ denote the ($k$-nice) graph obtained from $H^d_u$ by deleting all incoming edges of $x$, deleting all outgoing edges of $z$, and deleting all vertices that do not belong to a path of length at most~$d'$ from $x$ to $z$. 
We are now ready to compute $W^d_u$, for \mbox{$u \in S$},  recursively as described in~\cref{alg:weak}, where \textsc{VertexSeparator($H,x,y$)} is an arbitrary polynomial time algorithm computing a set of vertices separating $x$ and~$y$ in~$H$ (say the Edmonds-Karp algorithms running in time $\Oof(k\cdot (n+m))$). 

\begin{algorithm}[t]
\caption{Algorithm for computing weakly relevant vertices for $u \in S$}\label{alg:weak}
\begin{algorithmic}
\Procedure{\textsc{WeaklyRelevant}}{$G, S, u, d$}
\State $\textsf{return}~\textsc{Recurse}(H^d_{u,v}, u, v, \{\}, d)$ \Comment{Returns $W^d_u$}
\EndProcedure

\State ~

\Procedure{\textsc{Recurse}}{$H, x, z, W, d$}
\If{$xz\in E(H)$}
\State $\textsf{return}~W$
\EndIf
\If{$|V(H) \setminus \{x,z\}| \leq k$ \textbf{~or~} $d = 2$}
\State $\textsf{return}~W \cup (V(H) \setminus \{x,z\})$
\EndIf
\State $Y \gets \textsc{VertexSeparator}(H,x,z)$
\Comment If $xz\not\in E(G)$, then $|Y| \leq k$ and $x, z \not\in Y$
\State $W \gets W \cup Y$
\For{$y \in Y$}
\State $W \gets W \cup \textsc{Recurse}(H^{d - 1}_{x,y}, x, y, W, d - 1) \cup \textsc{Recurse}(H^{d - 1}_{y,z}, y, z, W, d - 1)$
\EndFor

\State $\textsf{return}~W$
\EndProcedure
\end{algorithmic}
\end{algorithm}

\medskip
\begin{lemma}\label{lem:subsume-paths}
For $u\in S$, every induced cycle $C_u$ of length at most~$d$ including~$u$ only includes vertices that are $d$-weakly relevant for~$u$, i.e., $V(C_u) \subseteq W^d_u$. 
\end{lemma}

\begin{proof}
We prove, by induction on $2 \leq \ell \leq d$ that every vertex of every induced $u$-$v$-path $P$ of length at most $\ell$ in $H^\ell_u$ is contained in $W^\ell_u$. Recall that every induced cycle of length at most~$d$ including $u$ in $G$ corresponds to an induced path of length at most~$d$ in $H^d_u$. Hence, all vertices of such induced cycles belong to $W^d_u$, proving the statement of the lemma. 

The claim is true for $\ell=2$; the only (induced) $u$-$v$-paths of length $\ell = 2$ in~$H^2_{u,v}$ involve at most $k$ distinct vertices by \cref{rule:disjoint-cycles}; otherwise $u$ belongs to $k+1$ $2$-cycles that pairwise intersect at $u$ and $u$ would be removed by~\cref{rule:loops}. These at most $k$ vertices belong to $W^2_{u}$. 
Now assume the claim is true for some $\ell \geq 2$. We prove it for~$\ell+1$. As $uv \not\in E(G)$, there exists an $u$-$v$-separator $Y$ of size at most~$k$. 
In particular, there is $y \in Y \cap V(P)$, say $y$ is the $j$-th vertex on $P$ when walking from $u$, $1 \leq j \leq \ell$. Then $P=P_1P_2$, where $P_1$ has length $1 \leq j \leq \ell$ and $P_2$ has length $\ell + 1 -j \leq \ell$. By the induction hypothesis, the vertices of $P_1$ are contained in $\textsc{Recurse}(H^{\ell}_{u,y}, u, y, W, \ell)$ and the vertices of $P_2$  are contained in $\textsc{Recurse}(H^{\ell}_{y,v}, y, v, W, \ell)$. 
By construction, $W^{\ell+1}_{u}$ contains the vertices of both of these sets, as needed to conclude the proof. 
\end{proof}


\cref{lem:subsume-paths} immediately implies the safeness of the following rule. 

\medskip
\begin{drule}\label{rule:disjoint-cyclesnew}
    If a vertex $w \not\in S$ is not $d$-weakly relevant for some vertex $u\in S$, then remove~$w$ from $G$. 
\end{drule}
\medskip

It remains to prove that the rule can be efficiently implemented and that its application leads to a small kernel. 

\medskip
\begin{lemma}\label{lem:kernelsize}
For $u \in S$ and $2 < \ell\leq d$ we have $|W^\ell_{u}| \leq k(2|W^{\ell-1}_{u}| + 1) \leq 2^{\ell - 1} k^{\ell - 1}$ (assuming $k \geq 1$).
\end{lemma}
\begin{proof}
The claim follows by induction. 
For $\ell=2$ we have $|W^2_{u}| \leq k$. 
By the recursive definition of $W^{\ell + 1}_{u}$ we have \mbox{$|W^{\ell + 1}_{u}|\leq k(2|W^{\ell-1}_{u}| + 1)\leq (2^\ell k^\ell - 2k)/(4k - 2)$} $\leq (2^\ell k^\ell)/(4 k - 2) \leq (2^\ell k^\ell)/(2k) \leq 2^{\ell - 1} k^{\ell-1}$.
\end{proof}

\begin{lemma}\label{lem:running-time-rule6}
    \cref{rule:disjoint-cyclesnew} is safe, and if $2^{d} k^{d} \leq n$, then it can be implemented in time $\Oof(n^2(n+m))$. 
\end{lemma}
\begin{proof}
For each $u\in S$ we compute the set $W^d_u$ by
applying \cref{alg:weak}. 
Each run of the algorithm requires $\Oof(2^{d-1} k^{d-1} \cdot k(n+m))=\Oof(2^{d-1}k^d(n+m))$ time in the worst case. Since $|S| \leq kd$, the total running time is $\Oof(2^{d} k^{d+1} (n+m))$ in the worst case. Hence, if $2^d k^{d} \leq n$ we have $2^{d} k^{d+1} (n+m) \leq n^2(n+m)$, as needed. 
\end{proof}


\begin{theorem}
\textsc{DFVS} parameterized by solution size~$k$ and restricted to graphs without induced cycles of length greater than $d$ admits a kernel with $2^d k^d$ vertices computable in time $\Oof(n^2(n+m))$.
\end{theorem}

\begin{proof}
Either $2^dk^d>n$, in which case we are done.  Otherwise, \cref{rule:disjoint-cyclesnew} is efficiently applicable and yields a kernel of the claimed size. 
\end{proof}

Finally, we further study the structure of kernelized instances and count how many induced cycles we can find. Our key tool is the classical sunflower lemma. 
A \emph{sunflower} with $\ell$ \emph{petals} and a \emph{core} $Y$ in a set system $\Gg$ is a 
collection of sets $S_1,\ldots, S_\ell\in E(\Gg)$ such that $S_i\cap S_j=Y$ for all $i\neq j \leq \ell$. 
The sets $S_i\setminus Y$ are called \emph{petals} and we require none of them to be empty (while the core $Y$ may be empty). 
Erd\"os and Rado~\cite{erdos1960intersection} proved in their famous sunflower lemma that every hypergraph with edges of size at most $d$ with at least $\mathrm{sun}(d,k)=d!k^d$ edges contains a sunflower with at least $k+1$ petals. 
Kernelization for $d$-\textsc{Hitting Set} based on the sunflower lemma yields a kernel with at most $\Oof(d!k^d)$ sets on hypergraphs with hyperedges of size at most~$d$, see e.g.~\cite{fafianie2015shortcut,van2014towards}.
We can prove the following lemma. 

\medskip
\begin{lemma}\label{thm:number-of-short-cycles}
Kernelized instances of \textsc{DFVS} contain at most $d^{3d}k^d$ induced cycles of length at most~$d$. 
\end{lemma}

\begin{proof}
We consider the hypergraph $\Gg$ of induced cycles. We prove that $\Gg$ does not contain a sunflower with more than $(d/2)^2k$ petals (each petal of size at most~$d$). Then, by the sunflower lemma, $\Gg$ contains at most $\mathrm{sun}(d, (d/2)^2k)=d!((d/2)^2k)^d=d!(d/2)^{2d}k^d\leq d^{3d}k^d$ hyperedges (of size at most $d$), as claimed.


A sunflower in $\Gg$ corresponds to a set of induced cycles in $G$ of length at most~$d$ that share a common core. 
If the core has the elements $\{v_1,\ldots, v_c\}$, then each petal~$S$ of the sunflower completes the vertices $v_1,\ldots, v_c$ to a cycle $C_S$. 
Observe that when two vertices $v_i, v_j$ are connected by an edge, say $v_iv_j\in E(G)$, then they appear in that order on every cycle $C_S$, and we have $v_jv_i\not\in E(G)$ as this would contradict the fact that the cycles are induced. 
Hence, the vertices $v_1,\ldots, v_c$ can be partitioned into maximal path segments $P_1,\ldots, P_t$ such that the vertices of each $P_i$ are connected consecutively as a path and such that there are no edges between $P_i$ and $P_j$ for $i\neq j$. 
Each of the cycles $C_S$ connects the path segments in some order using the petal vertices (the case $t = 1$ is the simplest so we consider the case $t > 1$). 
Note that if the path segments~$P_i$ and~$P_j$ are connected in that order, then the connection is via the last vertex of $P_i$ and the first vertex of $P_j$.  

Observe that $t\leq d/2$, as each cycle has length at most~$d$ and for every two consecutive path segments there must be at least one petal vertex connecting the two. 
Now, if there is a sunflower with more than $(d/2)^2k$ petals, then one of the possible $(d/2)^2$ pairs of path segments must be connected by more than $k$~paths. 
Since the connection is always between the last vertex $v$ of the first segment and the first vertex~$w$ of the second segment, there are more than~$k$ disjoint paths connecting $v$ and $w$. As \cref{rule:disjoint-cycles} can no longer be applied, there is a direct edge between $v$ and $w$, contradicting the fact that the cycles of the sunflower are induced. 
\end{proof}

Of course for small values of $d$ we can prove better bounds, however, they can improve the bounds of \cref{thm:number-of-short-cycles} only up to the constants depending on~$d$. 

\medskip
\begin{lemma}\label{thm:number-of-2-3-cycles}
If a kernelized instance of \textsc{DFVS} contains more than $k^2$ cycles of length~$2$ we may reject the instance. Furthermore, it contains at most $k^3$ induced cycles of length~$3$. 
\end{lemma}

\begin{proof} 
The first observation is based on the standard high-degree rule for \textsc{Vertex Cover}. 
No vertex $v$ is contained in more than $k$ cycles of length~$2$, otherwise \cref{rule:disjoint-cycles} applies and introduces a loop at $v$, which leads to the elimination of $v$ by \cref{rule:loops}. 
Hence, if we have more than $k^2$ cycles of length~$2$ these cannot be hit by a feedback vertex set of size $k$ and we can reject the instance. 

Assume there are more than $k^3$ induced cycles of length $3$. In any dfvs of size at most~$k$, there must  exist a vertex $v_1$ that hits a $1/k$ fraction of these cycles, i.e., $v_1$ must intersect with more than $k^2$ of the induced cycles of length~$3$. 
We fix such a $v_1$ and consider all induced $3$-cycles containing~$v_1$. 

For every in-neighbor $v_2$ of $v_1$, i.e., for every edge $v_2v_1 \in E(G)$, we can have at most $k$ (distinct) vertices $v_3$ such that $v_1v_3, v_3v_2 \in E(G)$; as otherwise, by \cref{rule:disjoint-cycles},  $v_1v_2 \in E(G)$ and the cycles are not induced. 
Similarly, for every out-neighbor~$v_2$ of~$v_1$, i.e., for every edge $v_1v_2 \in E(G)$, we can have at most $k$ (distinct) vertices~$v_3$ such that $v_2v_3, v_3v_1 \in E(G)$; as otherwise, again by \cref{rule:disjoint-cycles}, $v_2v_1 \in E(G)$ and the cycles are not induced. 

Since $v_1$ does not lie on $k+1$ cycles that pairwise intersect only at $v_1$ after the application of \cref{rule:disjoint-cycles}, by Menger's theorem there is a set of at most $k$ vertices (different from $v_1$)  that hits all cycles containing~$v_1$. Hence, at most $k$ (in or out) neighbors of~$v_1$ hit all induced $3$-cycles containing~$v_1$. 
Combined with the fact that each of those neighbors can belong to at most $k$ induced $3$-cycles containing~$v_1$, this implies that $v_1$ belongs to at most~$k^2$ induced cycles of length~$3$, contradicting the fact that $v_1$ must hit more than~$k^2$ induced $3$-cycles for a dfvs of size~$k$ to hit more than $k^3$ induced $3$-cycles.  
%
\end{proof}

In fact the bounds of \cref{thm:number-of-short-cycles} are optimal up to factors depending only on $d$. Consider for example the graph on vertices 
$v_{i,j}$ for $1\leq i\leq d, 1\leq j\leq k$. Connect $v_{i,j}$ with $v_{i+1, \ell}$, $1\leq i\leq d, 1\leq j,\ell\leq k$, where we compute $i+1$ modulo~$d$. This graph on $dk$ vertices has a dfvs of size~$k$. None of the presented reduction rules is applicable. Finally, it has $k^d$ cycles. 
\section{Nowhere dense classes without long induced cycles}

We now improve the general kernel construction for \textsc{DFVS} on graphs without induced cycles of length greater than $d$ by further restricting the class of (the underlying undirected) graphs. 
Classes with bounded expansion and nowhere dense classes are very general classes of sparse undirected graphs, including e.g.\ the class of planar graphs, all classes with excluded minors, excluded topological minors, and many more classes of uniformly sparse graphs. 
When we say that a directed graph $G$ belongs to a class~$\Cc$ of undirected graphs we in fact mean that the underlying undirected graph belongs to~$\Cc$. 
We prove that for every class~$\Cc$ with bounded expansion and for every $d\in \N$ there exists an algorithm and a constant $f_\Cc(d)$ such that 
for graphs $G\in \Cc$ without induced cycles of length greater than $d$ the algorithm computes a kernel with $f_\Cc(d)\cdot k$ vertices in time $f_\Cc(d)\cdot n^{\Oof(1)}$. 
For every nowhere dense class $\Cc$, every $d\in \N$, and every $\epsilon>0$ there exists and algorithm and a constant $f_\Cc(d,\epsilon)$ such that for graphs $G\in \Cc$ without induced cycles of length greater than $d$ the algorithm computes a kernel with $f_\Cc(d,\epsilon)\cdot k^{1+\epsilon}$ vertices in time $f_\Cc(d,\epsilon)\cdot n^{\Oof(1)}$. 
Recall that the constants in the $\Oof$-notation do not depend on $d$ or~$\epsilon$. 

We present the proof for nowhere dense classes, as every class of bounded expansion is also nowhere dense. To keep the presentation clean we omit the details for the former case since the required modifications are negligible.   
We refer the reader to \cite{nevsetvril2008grad, NesetrilM11a} for formal definitions of bounded expansion and nowhere dense classes of graphs. 

We collect some additional properties that we need for our additional reduction rule.
The first property we use is that for every nowhere dense class of graphs $\Cc$ there exists a positive integer $t > 0$ such that $K_{t,t}$ (the complete bipartite graph with $t$ vertices in each part) is not a subgraph of any $G\in \Cc$. 
We remark that here and in the lemmas that we state below, the appearing constants (like $t$) can either be computed or approximated by uniform algorithms that do not depend on the graph class. 
It may also be the case that we do not even need to compute the constants and only use their existence to derive size bounds on the resulting kernelized instances.
When we define additional constants from constants whose existence is stated in the lemmas, it is easy to verify that these constants are also computable, as they are obtained by simple arithmetic operations from these computable constants. 

Let us fix an approximate solution $S$ as described in  \cref{lem:approx-dfvs}. 
We construct a so-called projection closure around our approximate solution $S$. This is possible in nowhere dense classes as stated in the next lemma.
In the following, when we speak of an undirected path $P$ in $G$, we mean an undirected path in the underlying undirected graph. 

Let $X\subseteq V(G)$ and let $u\in V(G)\setminus X$. The \emph{undirected $d$-projection} of $u$ onto~$X$ is defined as the set $\Pi_d(u,X)$ of all vertices $w\in X$ for which 
there exists an undirected path $P$ in $G$ that starts in $u$, ends in $w$, has length at most $d$, and whose internal vertices do not belong to $X$. 

\newcommand{\Xclos}{\ensuremath{X^\circ}}
\newcommand{\Sclos}{\ensuremath{S^\circ}}

\medskip
\begin{lemma}[\cite{EickmeyerGKKPRS17}]\label{lem:closure-nd}
    Let $\Cc$ be a nowhere dense class of graphs. There exists a polynomial time algorithm that given a graph $G\in \Cc$, $d$, $\epsilon>0$ and $X\subseteq V(G)$, computes the $d$-projection-closure of $X$, denoted by $\Xclos$, with the following properties:\\[-3mm]
    \begin{enumerate}
        \item $X\subseteq \Xclos$, \\[-3mm]
        \item $|\Xclos|\leq \kappa_{d,\epsilon}\cdot |X|^{1+\epsilon}$ for a constant $\kappa_{d,\epsilon}$ depending only on $d$ and $\epsilon$, \\[-3mm]
        \item $|\Pi_d(u,\Xclos)|\leq \kappa_{d,\epsilon}\cdot |X|^\epsilon$ for each $u\in V(G)\setminus \Xclos$, and\\[-3mm]
        \item $|\{\Pi_d(u,\Xclos)~:~u\in V(G)\setminus \Xclos\}|\leq \kappa_{d,\epsilon}\cdot |X|^{1+\epsilon}$. \\[-3mm]
\end{enumerate}
\end{lemma}

\medskip
In words, the projection closure $\Xclos$ extends $X$, it is not much larger than $X$, all undirected $d$-projections to $\Xclos$ are really small, and finally, there are not many different projections onto $\Xclos$. 

We need the following strengthening for $\ell$-tuples~\cite{PilipczukST18a}. 
For a set $X\subseteq V(G)$ and an $\ell$-tuple $\bar x$ of vertices we call the tuple $(N[\bar x_1]\cap X, \ldots, N[\bar x_\ell]\cap X)$ the \emph{undirected projection of $\bar x$ onto $X$}. 
We say that $\bar X = (X_1,\ldots, X_\ell)$ is realized as a projection if there is a tuple $\bar x$ whose projection is equal to $\bar X$.

\medskip
\begin{lemma}[\cite{PilipczukST18a}]\label{lem:number-of-types}
Let $\Cc$ be a nowhere dense class of graphs and let~$\ell\geq 1$. Let $G\in \Cc$ and $X\subseteq V(G)$. Then, for every $\epsilon>0$ there exists a constant~$\tau_{\ell,\epsilon}$ such that there are at most $\tau_{\ell,\epsilon}\cdot|X|^{\ell+\epsilon}$ different realized undirected projections of $\ell$-tuples. 
\end{lemma}
\medskip

Let $X\subseteq V(G)$ and let $x,y\in X$. Let $P=u_1,\ldots, u_\ell$ be an (almost) induced $x$-$y$-path with $|V(P)| = \ell \leq d$ and let $u_i \in V(P)$. Then, the \emph{$X$-path-projection profile} of $(P,u_i)$ is the tuple $(i, N^-(u_1)\cap X, N^+(u_1)\cap X,\ldots, N^-(u_\ell)\cap X, N^+(u_\ell)\cap X)$. 
The \emph{$X$-path-projection profile} of a vertex $u$ is the set of all $X$-path-projection profiles $(P,u)$, where~$P$ is any almost induced $x$-$y$-path on at most~$d$ vertices and $x,y\in X$ are any two vertices in~$X$.  
Two vertices $u,v$ are \emph{equivalent over $X$} if they have the same $X$-path-projection profiles.

\medskip
\begin{lemma}\label{lem:num-profiles}
Let $\Cc$ be a nowhere dense class of graphs and let $t > 0$ be some fixed positive integer such that $K_{t,t}\not \subseteq G$, for all $G\in \Cc$. Let $G\in \Cc$ and \mbox{$X\subseteq V(G)$}. 
Then, for every $\epsilon>0$, there exists a constant $\chi_{d,t,\epsilon}$ such that the number of $X$-path-projection profiles for $u\in V(G)\setminus X$ is bounded by $\chi_{d,t,\epsilon}\cdot |X|^{d+\epsilon}$. 
\end{lemma}

\begin{proof}
We have $d$ choices for the number $i$. 
By \cref{lem:number-of-types} we have at most $\tau_{\ell,\epsilon}|X|^{\ell+\epsilon}$ different undirected projections of $\ell$-tuples. 
If the undirected projection of a single vertex $v$ within an $\ell$-tuple has size smaller than $t$ \mbox{($|N[v]\cap X| \leq t - 1$)}, then even though we can have many vertices with the same undirected projection, there are at most $2^{t-1}$ possible ways of orienting this undirected projection to obtain a directed projection; orienting an undirected projection $N[v]\cap X$ yields a directed projection $(N^-[v]\cap X, N^+[v]\cap X$). 
Otherwise, when $|N[v]\cap X| \geq t$, we can have at most $t-1$ other vertices with the same undirected  projection; this follows from the fact that $G$ does not contain $K_{t,t}$ as a subgraph. 
Consequently, there are at most $t-1$ possible directed projections with the undirected projection $N[v]\cap X$ whenever $|N[v]\cap X| \geq t$. 
Putting it all together, we know that any undirected projection (of any size) can be oriented in at most $2^{t - 1}$ different ways.  
Summing over all possible choices of $\ell\leq d$, we get at most $d^2\cdot 2^{dt} \cdot \tau_{d,\epsilon} \cdot |X|^{d+\epsilon}$ $X$-path-projection profiles. To conclude the proof, we define $\chi_{d,t,\epsilon}$ as $d^2\cdot 2^{dt} \cdot \tau_{d,\epsilon}$.
\end{proof}


We now state a new reduction rule, which depends on a constant $c$ that we fix later. 
We remark that $c$ can be computed from a given graph and $d$, however, it depends on other constants defined in the next lemma, so that its definition at this point would not be comprehensible. 

\medskip
\begin{drule}\label{rule:BE}
Assume we can find in polynomial time sets $B,X\subseteq V(G)$ such that \\[-3mm]
    \begin{enumerate}
        \item the $d$-neighborhoods in $G-X$ of distinct vertices from $B$ are disjoint, \\[-3mm]
        \item every induced cycle using a vertex of $N^d_G[B]$ also uses a vertex of $X$, \\[-3mm]
        \item vertices in $B$ are pairwise equivalent over $X$, i.e., they have the same $X$-path-projection profile (in particular, if one vertex of $B$ lies on an $x$-$y$-path for $x,y\in X$ of length $\ell \leq d$, then all vertices of $B$ do as well), and\\[-3mm]
        \item $|B|>c+d+1$ and $|X| \leq c$. \\[-3mm]
    \end{enumerate}
    Then, choose an arbitrary vertex of $B$ and delete it from $G$.
\end{drule}
\medskip

\begin{lemma}
    \cref{rule:BE} is safe. 
\end{lemma}
\begin{proof}
    Denote by $G'$ the graph obtained after one application of the rule with sets $B,X \subseteq V(G)$, where $|B|>c+d+1$ and $|X| \leq c$. Let $u \in B$ be the deleted vertex. 
    Since the rule only removes a vertex it is clear that every dfvs $S$ of $G$ is also a dfvs of~$G'$. It remains to show that for every dfvs $S'$ of $G'$ there exists a dfvs of $G$ that is not larger than $S'$. 

    Let $C_u$ be an induced cycle in $G$ going through $u$ such that $|V(C_u)| \leq d$. All other induced cycles (that do not contain $u$) are also induced cycles in $G'$ and are hence hit by~$S'$.  
    By assumption, every induced cycle including a vertex of $N^d_G[B]$, in particular~$u$, also includes at least one vertex of $X$. Pick $x,y\in X$ (possibly $x=y$) such that $C_u$ includes $x,u,y$ in that order and such that no other vertices of $X$ appear in between~$x$ and~$y$ (in $C_u$).  

    Fix the $x$-$u$-path $P_{xu}$ and the $u$-$y$-path $P_{uy}$ that are subpaths of $C_u$. Let $P_u=P_{xu}P_{uy}$. 
    Since $C_u$ is an induced cycle, if $C_u$ contains  vertices of $X \setminus \{x,y\}$, then these other vertices do not appear in the $X$-path-projection profile of $(P_u, u)$. 

    Since all $v\in B$ have the same $X$-path-projection profile, for each $v \neq u$ there are paths~$P_{xv}$ from $x$ to $v$ and $P_{vy}$ from $v$ to $y$ such that $P_v=P_{xv}P_{vy}$ and $(P_v, v)$ has the same $X$-path-projection profile as $(P_u, u)$. 
    Since the $d$-neighborhoods in $G-X$ of all vertices of $B$ are pairwise disjoint, all of these paths are pairwise vertex disjoint except for $x$ and~$y$, and at most $d$ vertices $v\in B\setminus\{u\}$ can have other vertices of $V(C_u) \setminus X$ that are in the $1$-neighborhood of $P_v$. 
    Since all these pairs $(P_v, v)$, $v \neq u$, have the same $X$-path-projection profile as $(P_u, u)$, just like $(P_u, u)$, vertices in $(V(C_u) \cap X) \setminus \{x,y\}$ do not appear in the $X$-path-projection profile of $(P_v, v)$.
    Because there are at least $c+d+1$ vertices in $B\setminus\{u\}$, we get $c+1$ induced $x$-$y$-paths that are pairwise vertex disjoint except for $x$ and~$y$ in $G'$ and that are not adjacent to vertices of $V(C_u) \setminus \{x,y\}$. Hence, for every $v\in B$ we get an induced cycle $C_v$ on at most $d$ vertices by replacing~$P_u$ by~$P_v$.
    Denote the set of these vertices $v$ by $B'$. 

    Assuming that $S'$ is not a dfvs in $G$, the vertices in $V(C_u) \setminus V(P_u)$ are not hit by~$S'$. 
    Then, all the (at least $c+1$) cycles $C_v$ for $v\in B'$ are hit on the paths $P_v$, i.e., $S' \cap V(C_v) \subseteq V(P_v)$. 
    All vertices of $S'$ that hit a cycle $C_v$ on $P_v$, call those vertices~$Z_v$, lie in the $d$-neighborhood of~$v$ in $G-X$. 
    Hence, by assumption, they do not hit any cycles that do not also go through~$X$. 
    Let $Z=\bigcup_{v\in B'} Z_v$. 
    Then, $(S'\setminus Z)\cup X$ is a dfvs of $G$ of size at most $|S'|$; as $|Z| \geq c + 1$ and $|X| \leq c$.
\end{proof}

\begin{lemma}
    There exists a function $g_\Cc$ such that 
    given a graph $G$, $X\subseteq V(G)$ and $u\in V(G)$, we can test in time $g_\Cc(d,\epsilon)\cdot n^{\Oof(1)}$ whether every induced cycle (on at most $d$~vertices) including a vertex of~$N^d_G[u]$ also includes a vertex of $X$. Furthermore, we can decide in the same time if two vertices $u,v$ have the same $X$-path-projection profile.
\end{lemma}
\begin{proof}
    We apply the efficient first-order model checking algorithm of \cite{grohe2017deciding}, which yields the desired running time whenever we test a formula of length depending only on $d$. To do so, we first mark the set $X$ using a unary predicate to make it accessible to first-order logic. 
    Both properties are easily expressible by a first-order formula whose length depends only on $d$. 
    This is cumbersome, but let us demonstrate how to write a formula that expresses that all induced cycles of length at most $d$ containing a vertex of $N_G^d[u]$ also contain a vertex of $X$. 
    First, we write a formula expressing that the distance between two vertices is at most $d$. 
    \[dist_{\leq d}(x,y)=\exists x_1\ldots \exists x_{d+1}\big(x=x_1\wedge y=x_{d+1} \wedge \bigwedge_{1\leq i\leq d}(E(x_i,x_{i+1})\vee (x_i=x_{i+1}))\big)\]
    
    The formula quantifies the at most $d+1$ vertices on a path of length at most $d$. It ensures that the first vertex is $x$ and the last vertex is $y$. 
    It then states that between consecutive vertices there must be an edge, or the vertices are equal (accounting for the possibility that the path is shorter). 

    Similarly, we can speak about the induced cycles of length at most $d$. We additionally state that if a vertex of the cycle is at distance at most $d$ to the vertex $u$, then it shall include a vertex of $X$. 

    \begin{align*}
        \phi(u,X) = & \exists x_1\ldots \exists x_d \big(x_1\neq x_d \wedge \bigwedge_{1\leq i\leq d}(E(x_i,x_{i+1})\vee (x_i=x_{i+1})) \wedge E(x_d,x_1) \wedge \\
        & \bigwedge_{1\leq i\leq j+2\leq d} \neg E(x_i, x_{j+2})\wedge \bigwedge_{d> j>i\geq 1} \neg E(x_j,x_i)\wedge \\
        & \big (\bigvee_{1\leq i\leq d} dist_{\leq d}(x_i,u)\big) \rightarrow \big( \bigvee_{1\leq i\leq d}X(x_i)\big) \big)
    \end{align*}

The formula quantifies the vertices of the cycle, where it allows all vertices to be equal, except for $x_1$ and $x_d$. 
By this, we enforce that the cycle has length at least $2$ (unlike for the distance formula, where we allow paths of length $0$). 
We then state that the cycle must be induced by stating that there are no edges from a vertex to another vertex of the cycle that is at least two steps away from the first vertex and that there are no back-edges on the cycle. 
Finally, we state that if some vertex of the cycle is at distance at most $d$ to $u$, then some vertex must also be in $X$. 
    
The second property can be stated similarly. We explicitly quantify the vertices of an (almost) induced path containing $u$ and an (almost) induced path containing $v$ such that~$u$ and $v$ are at the same position in the path and that all vertices in position~$j$ of the two paths have the same projections on $X$. 
To state the latter property for two vertices~$x$ and $y$ consider the formula $\forall z \big(X(z)\rightarrow (E(x,z)\leftrightarrow E(y,z)\wedge E(z,x)\leftrightarrow E(z,y))\big)$. 
\end{proof}

Nowhere dense classes of graphs are \emph{uniformly quasi-wide}~\cite{NesetrilM11a}. 
A class of graphs is uniformly quasi-wide if for every $q$ there exists a constant $s$ and a function $N:\N\rightarrow \N$ such that the following holds. 
For every $m\in \N$, if $G\in \Cc$ and $A\subseteq V(G)$ has size at least $N(m)$, then there exists a set $Y\subseteq V(G)$ of size at most $s$ and a set $B\subseteq A\setminus Y$ of size at least $m$ such that all distinct vertices of $B$ have disjoint $q$-neighborhoods in~$G-Y$. 
The best bounds for the function $N$ are given in \cite{KreutzerRS19, PilipczukST18a}. It is important that the function is polynomial, that is, $N(m)=m^t$ for some constant $t$ and that given $G$ and $A\subseteq V(G)$, the sets $S$ and $B$ can be efficiently computed. That is, they can be computed in time $u_\Cc(q)\cdot n^{\Oof(1)}$ for some function $u_\Cc$. 

\medskip
\begin{theorem}
There is a function $h_\Cc$ such that for any $\epsilon>0$
\cref{rule:BE} can be applied in time $h_\Cc(d,\epsilon)\cdot n^{\Oof(1)}$ until the reduced graph has at most $h_\Cc(d,\epsilon)\cdot k^{1+\epsilon}$ vertices. 
\end{theorem}

\begin{proof}
Let $N(m)$ and $s$ be the function and constant witnessing that $\Cc$ is uniformly quasi-wide for parameter $q = d$. Assume $N(m) = m^t$. 
Recall that~$S$ is an approximate solution of size at most $dk$. Let $\delta>0$ be a constant that we determine later. 
We build the projection closure $\Sclos\supseteq S$ for parameters $d$ and~$\delta$, which by \cref{lem:closure-nd} is of size at most $\kappa_{d,\delta}\cdot (dk)^{1+\delta}$ and such that the $2d$-projection of each $v\in V(G)\setminus \Sclos$ has size at most $\kappa_{d,\delta}\cdot (dk)^\delta$. 
Let $\chi_{2d,t,\delta}$ be the constant from \cref{lem:num-profiles}. 
Define $c$ for the application of \cref{rule:BE} as $\kappa_{d,\delta}\cdot (dk)^\delta+s$. 
    
    Assume $|V(G)|>\kappa_{d,\delta}(dk)^{1+\delta}+\kappa_{d,\delta}\cdot N(\chi_{d,t,\delta}\cdot c^{d+\delta}$ \mbox{$\cdot(c+d+2))(dk)^{1+\delta}$}. We show that we can efficiently apply \cref{rule:BE}.
    
    First, there are more than $\kappa_{d,\delta}\cdot N(\chi_{d,t,\delta}\cdot(c+d+2))(dk)^{1+\delta}$ vertices in $V(G)\setminus \Sclos$. 
    Moreover, every induced cycle using a vertex $u\in V(G)\setminus \Sclos$ uses a vertex of $\Pi_d(u)\subseteq \Sclos$. 
    This is true because $G - S$ is acyclic, induced cycles have length at most~$d$, and all paths of length at most $d$ from $u$ to~$S$ must use a vertex of $\Pi_{d}(u)$. 

    Since by \cref{lem:closure-nd} there are at most $\kappa_{d,\delta} \cdot (dk)^{1+\delta}$ different projection classes to~$\Sclos$, there is at least one class~$A$ with at least $N(\chi_{d,t,\delta}\cdot c^{d+\delta}\cdot(c+d+2))$ vertices. 
    We denote the set of vertices in the projection of the vertices from this projection class by $\Pi \subseteq \Sclos$. 
    We compute a set $Y$ of size at most $s$ and a set $B'\subseteq A\setminus Y$ containing at least  $\chi_{d,t,\delta}\cdot c^{d+\delta}\cdot(c+d+2)$ vertices that have pairwise disjoint $2d$-neighborhoods in $G-Y$. This is possible as $G$ comes from a uniform quasi-wide class with the function~$N$ and constant $s$ as fixed above. 
    Let $X=\Pi\cup Y$. Note that $X$ has size at most $\kappa_{d,\delta}\cdot (dk)^\delta+s=c$.
    
    By \cref{lem:num-profiles} there are at most $\chi_{d,t,\delta}\cdot c^{d+\delta}$ many different $X$-path-projection profiles, hence, we find a set $B\subseteq B'$ of size greater than $c+d+1$ of vertices that all have the same $X$-path-projection profile. Hence, all assumptions to apply \cref{rule:BE} are satisfied and we can still carry out the rule to decrease the size of~$V(G)$. 
    As all lemmas can be applied efficiently,  \cref{rule:BE} with the given sets $B$ and~$X$ can also be applied efficiently. 

    It remains to define the constant $\delta$. We need $N(\chi_{d,t,\delta}\cdot c^{d+\delta}\cdot(c+d+2))(dk)^{1+\delta}\leq  h_\Cc({d,\epsilon})\cdot k^{1+\epsilon}$ for some function $h_\Cc$. 
    By assumption we have $N(m)=m^t$, hence, 
    \begin{align*}    
     N(\chi_{d,t,\delta}& \cdot c^{d+\delta}\cdot(c+d+2))(dk)^{1+\delta}  \leq (\chi_{d,t,\delta}\cdot c^{d+\delta}\cdot(c+d+2))^t(dk)^{1+\delta}\\
    & \leq (\chi_{d,t,\delta}\cdot (\kappa_{d,\delta}\cdot (dk)^\delta+s)^{d+\delta}\cdot((\kappa_{d,\delta}\cdot (dk)^\delta+s)+d+2))^t(dk)^{1+\delta}\\
    & \leq \chi_{d,t,\delta}\cdot \kappa_{d,\delta}^{d+t+\delta}\cdot (2sd)^{d+t+\delta} \cdot (dk)^{d\delta+\delta^2+t\delta+1+\delta}
    \end{align*}
    It hence suffice to define $\delta$ such that $d\delta+\delta^2+t\delta+1+\delta\leq 1+\epsilon$, which is the case whenever $(d+t+1)\delta+\delta^2\leq \epsilon$. 
    We can choose $\delta=\epsilon/(d+t+2)$.
\end{proof}

\section{DFVS in planar graphs without long cycles}

One may wonder whether the stronger assumption that a graph does not contain long cycles, induced or non-induced, leads to even more efficient algorithms. 
We show that this is indeed the case when considering planar graphs. We show that strongly connected planar graphs without cycles of length $d$ have treewidth $\Oof(d)$. 
By \cref{lem:strong-components} weak components are equal to strong components. Hence, if each strong component has bounded treewidth the whole graph has bounded treewidth. 
We can then use the algorithm of Bonamy et al.~\cite{bonamy2018directed} to solve the instance in time $2^{\Oof(d)}\cdot n^{\Oof(1)}$. 

\medskip
\begin{lemma}\label{lem:stiching}
    Let $G$ be a strongly connected graph and let $u,v\in V(G)$. Let $P$ be an (undirected) path between $u$ and $v$ in the underlying undirected graph. 
    If $G$ does not have cycles of length greater than $d$, then it contains a (directed) $u$-$v$-path $Q$ such that every vertex of $Q$ is at (directed) distance at most $d$ from some vertex of $P$.
\end{lemma}
\begin{proof}
    Assume $P=v_1\ldots v_\ell$. For each $v_iv_{i+1}$ fix a cycle $C_i$ of length at most $d$ containing both $v_i$ and $v_{i+1}$. Such a cycle must exist as $G$ is strongly connected and does not contain cycles longer than $d$. 
    We can now appropriately stitch subpaths from these cycles to find a $u$-$v$-path $Q$ in $G$. It is immediate that every vertex of $Q$ is at distance at most $d$ from some vertex of $P$. 
\end{proof}

\begin{theorem}
Let $G$ be a strongly connected planar graph without cycles of length greater than $d$. Then, $G$ has treewidth at most $30d$. 
\end{theorem}

\begin{proof}
As proved in \cite{robertson1994quickly}, every planar graph of treewidth at least $6t$ contains a $t \times t$ grid as a minor. 
Assume towards a contradiction that $G$ has treewidth greater than~$30d$. Then it contains a grid of order $5d$ as a minor. 
We fix four vertices $v_1,v_2,v_3,v_4$ from the four corner branch sets of the inner sub-grid minor of order $3d$ and undirected paths $P_{i,(i \mod 4)+1}$ leading from $v_i$ to $v_{(i\mod 4)+1}$ using only vertices of branch sets of the boundary of the grid minor model. 
By \cref{lem:stiching} and because $G$ is planar we find $v_i$-$v_{(i\mod 4)+1}$-paths that use only vertices inside the regions defined by the boundary of the grid minor of order $5d$ and the boundary of the central sub-grid minor of order~$d$. 
By gluing the paths we find a closed walk which contains a cycle that fully encloses the central grid minor of order $d$. This cycle has length at least $4d$, contradicting the fact that $G$ has no cycles longer than $d$. 
\end{proof}
\section{A kernel with respect to the undirected feedback vertex set number}\label{sec:new-red-rules}

In this section we prove that \cref{rule:shortcut-new}, \cref{rule:loops}, and \cref{rule:disjoint-cycles} lead to a kernel of size $\Oof(f^4)$, where~$f$ is the size of a minimum  feedback vertex set in the underlying undirected graph. In fact, we prove the stronger bound of~$\Oof(kf^3)$. 
Our analysis is based on the analysis of Bergougnoux et al.~\cite{bergougnoux2021towards}. 
Essentially, we prove that all complicated rules of Bergougnoux et al.\xspace are subsumed by \cref{rule:shortcut-new}. 

In the following, we fix an undirected feedback vertex set~$F$ (which does not have to be computed and in particular may be assumed to be minimum). 
We prove that the rules lead to a kernel of size $\Oof(|F|^3k)$. We (almost) follow the terminology of Bergougnoux et al. For the sake of clarity, when it suffices to apply a special case of one of the above rules, we refer to the special case. 
Our analysis follows the same structure as that of Bergougnoux et al.

It will be convenient to partition the edge set of~$G$ into red and blue edges. 
We say that an edge $uv \in E(G)$ is \emph{red} if $vu \notin E(G)$, and say that $uv$ is \emph{blue} if $vu \in E(G)$, i.e., if $uvu$ is a cycle in $G$. 
By $R(G) \subseteq E(G)$, we denote the set of all red edges of $G$, and by $B(G) \subseteq E(G)$ we denote the set of all blue edges of $G$. 
We call a vertex blue if it is incident to at least one blue edge and denote by $B$ the set of blue vertices We call all other vertices red. 
By \cref{thm:number-of-2-3-cycles}, we may assume that there are at most $k^2$ blue edges, hence, $|B|\leq 2k^2$ (otherwise we have a negative instance). 
Remarkably, this is the only place where we use \cref{rule:loops}. 
Note that the standard argument applicable for \textsc{Vertex Cover} that we may assume that there are at most~$k^2$~blue vertices cannot be applied here, since there can be additional red edges. 
Let $A = V(G)\setminus (F\cup B)$. 


Observe that no graph $G$ has an induced cycle containing both a blue and a red edge, as any blue edge $uv$ already implies a cycle of length two. Hence, as every blue edge lies on an induced cycle by definition, it suffices to check whether red edges lie on induced cycles or not.

\medskip
\begin{observation}\label{lem:all-induced-cycles}
After exhaustive application of \cref{rule:shortcut-new} every vertex $v$ lies on a cycle that is induced around $v$ and on at least two cycles that are vertex disjoint except for~$v$.
In particular, every red vertex has only red neighbors and at least two in-coming and at least two out-going red edges.
\end{observation}


\subsection{Analysis of kernel size}\label{sec:bergougnoux}

As mentioned, our analysis follows the same structure as that of Bergougnoux et al. 
An ordered pair $(u,v)$ of (not necessarily distinct) vertices of $F$ is called a \emph{potential edge} of~$F$. If $uv\not\in E(G)$, then $(u,v)$ is a \emph{non-edge} of $F$. 
If $u=v$, then $(u,v)$ is called a \emph{loop}. 
A vertex $w\in A$ \emph{directly contributes} to a potential edge $(u,v)$ if $(u,w) \in E(G)$ and $(w,v) \in E(G)$. 

Note that, unlike the approach of Bergougnoux et al., no vertex of $A$ can directly contribute to a loop (the case $u=v$), these vertices are incident with a blue edge and have been collected in~$B$ (hence, do not belong to $A$). 
The following lemma follows from the fact that \cref{rule:disjoint-cycles} is no longer applicable.

\medskip
\begin{lemma}
\label{lem:direct-contributions}
For every non-edge $(u,v)$ of $F$ there are at most $k$ vertices that 
directly contribute to $(u,v)$. Consequently, there are at most $|F|(|F|-1)k$ vertices in $A$ that directly contribute to some non-edge of $F$.
\end{lemma}
\medskip

We follow the approach of Bergougnoux et al.\ to bound the number of vertices in~$A$. We denote by $A_0$, $A_1$, $A_2$ and $A_{\geq 3}$ the sets of vertices of $A$ that have a total degree $0$, $1$, $2$ and at least $3$, respectively, in $G[A]$. 

\medskip
\begin{lemma}\label{lem:non-edge-contributions}
Every vertex in $A_0\cup A_1$ directly contributes to a non-edge of~$F$.
\end{lemma}
\begin{proof}
By \cref{lem:all-induced-cycles}, every vertex of $A$ has at least~$2$~in- and $2$ out-neighbors, hence, every $w\in A_0\cup A_1$ has an in-neighbor and an out-neighbor in $F$. 
Assume that~$w$ does not directly contribute to a non-edge. 
Then for every pair of in- and out-edges $uw\in E(G)$ and $wv\in E(G)$ for $u,v\in F$ we have $uv\in E(G)$ ($u\neq v$ as otherwise $w\in B$). 
If $w$ has no in-neighbor from~$A$, then it can have at most one out-neighbor in~$A$. Call this out-neighbor $z$. 
Then every cycle that is induced around $w$ also contains~$z$, as we can only enter $w$ via some in-neighbor $u \in F$ and continue the cycle via $z$, because all other out-neighbors $v\in F$ are directly connected, hence, do not lie on a cycle that is induced around $w$. 
Hence, \cref{rule:shortcut-new} would have removed $w$. See \cref{fig:non-edge-contributions} for an illustration.

\begin{figure}
    \centering
    \begin{tikzpicture}
        \fill (2, 1) circle (2pt) node (w) {};
        \node[above=0mm of w] {$w$};

        \fill (4, 1) circle (2pt) node (z) {};
        \node[above=0mm of z] {$z$};

        \fill (0, -1)  circle (2pt) node (u1) {};
        \fill (0.5, 0) circle (2pt) node (u2) {};
        \fill (3, -1)  circle (2pt) node (v1) {};
        \fill (3.5, 0) circle (2pt) node (v2) {};

        \path[-{Latex}]
        (u1) edge (v1)
        (u1) edge (v2)
        (u2) edge (v1)
        (u2) edge (v2)

        (u1) edge (w)
        (u2) edge (w)
        (w) edge (v1)
        (w) edge (v2)

        (w) edge (z)
        ;

        \draw[dashed, rounded corners=5,gray!50!black] (-0.5,-1.5) rectangle ++(4.5, 2) node[below right=1.5 and 0] {$F$};
    \end{tikzpicture}
    \caption{Illustration of \cref{lem:non-edge-contributions}.}
    \label{fig:non-edge-contributions}
\end{figure}


Analogously, if~$w$ has no out-neighbor in~$A$, then it can have at most one in-neighbor from~$A$, call it again $z$. 
Then all cycles that are induced around $w$ must have~$z$ as the predecessor of~$w$, as the cycles using another in-neighbor $u \in F$ and then an out-neighbor $v \in F$ are not induced around $w$. 
Again $w$ would have been removed by \cref{rule:shortcut-new}. 
\end{proof}

\begin{corollary}\label{cor-013}
$|A_0\cup A_1|\leq |F|(|F|-1)k$ and $|A_{\geq 3}|\leq |F|(|F|-1)k-2$.
\end{corollary}

\begin{proof}
The underlying undirected graph of $G[A]$ induces an undirected forest and the number of vertices of degree at least $3$ in an undirected forest is at most equal to the number of leaves minus two; the number of leaves is $|A_1|$. Hence, we have $|A_{\geq 3}|\leq |A_1| - 2 \leq |A_0 \cup A_1| -2 \leq |F|(|F|-1)k-2$,  where the latter inequality is a consequence of the previous two lemmas (\cref{lem:direct-contributions} and~\cref{lem:non-edge-contributions}).
\end{proof}

It remains to bound the size of $A_2$. We call a vertex $w\in A_2$ a \emph{sink} vertex or a \emph{source} vertex if the two neighbors of $w$ in $A$ are both in-neighbors or out-neighbors, respectively. 
Otherwise we call~$w$ a \emph{balanced} vertex. As \cref{rule:Y-shape} cannot be applied, every vertex of~$A_2$ has at least $2$ distinct neighbors in $F$. 

Let $P = (w_1,\ldots,w_r)$ be an inclusion-wise maximal directed path in $G[A]$ whose internal vertices are in $A_2$. We call $P$ a \emph{path segment} in $A$. We call $P$ an \emph{outer path segment} if at least one of its endpoints is not in $A_2$ and an \emph{inner path segment}, otherwise. Note that path segments are directed paths, which, by maximality, can never start or end with a balanced vertex. Moreover, every internal vertex of a path segment must be a balanced vertex. 

We first bound the number of outer path segments, that is, the number of path segments with at least one endpoint in $A_1\cup A_3$. 

\medskip
\begin{lemma}
The number of outer path segments is at most $4(|F|(|F|-1)k)$. 
\end{lemma}
\begin{proof}
    Let $H$ be the undirected graph obtained from the undirected graph underlying~$G[A]$ by contracting all edges that are incident to at least one vertex of degree~$2$. 
    Then every outer path segment runs between the endpoint of an edge in $H$ and an endpoint or inner vertex of the contracted path in a unique direction, as $A$ does not contain blue edges. 
    Hence, the number of outer path segments in $A$ is bounded by twice the number of edges of $H$. 
    As $H$ is a forest without vertices of degree two, its number of edges is equal to the number of leaves plus the number of non-leaves minus one. As shown in \cref{cor-013} both of these numbers are bounded from above by $|F|(|F|-1)k$, as needed. 
\end{proof}

We say that a path segment $P = (w_1,\ldots,w_r)$ \emph{contributes to a potential edge} $(u,v)$ of $F$ if there are $i$ and $j$, $1\leq i\leq j\leq r$, such that $uw_i\in E(G)$ and $w_jv\in E(G)$. We say that $P$ \emph{contributes to a loop} on $u\in F$ if there are $i$ and $j$, $1\leq i\leq j\leq r$, such that $uw_i\in E(G)$ and $w_ju\in E(G)$. 

\medskip
\begin{lemma}\label{lem:short-path-segments}
    If an inner path segment does not contribute to a non-edge or a loop of~$F$, then it has length at most $1$, that is, it consists of only $2$ vertices.
    
\end{lemma}
\begin{proof}
    If $P=(w_1,w_2,w_3,\ldots,w_r)$ for $r\geq 3$ does not contribute to a non-edge or loop, then every in-neighbor $u\in F$ of $w_1$ or $w_2$ is connected to every out-neighbor $v\in F$ of~$w_2$. 
    Then every cycle that is induced around $w_2$ comes from $w_1$ (in which case it comes from an in-neighbor $u$ of $w_1$) or from an in-neighbor $u\in F$ of~$w_2$. 
    The cycle cannot continue via an out-neighbor $v\in F$ of $w_2$, as this would not be a cycle induced around $w_2$ (as in any case we have the edge $uv\in E(G)$). 
    Hence, it must continue via~$w_3$. 
    That is, every cycle induced around $w_2$ also contains $w_3$, hence, $w_2$ is removed and shortcut by \cref{rule:shortcut-new}. See \cref{fig:short-path-segments} for an illustration.
    \begin{figure}
        \centering
        \begin{tikzpicture}
            \foreach \i in {1, ..., 6} {
              \fill (\i, 0) circle (2pt) node (\i) {};
            }
            \path[-{Latex}]
            (2) edge (1)
            (2) edge (3)
            (3) edge (4)
            (4) edge (5)
            (6) edge (5)
            ;

            \node[above=0mm of 2] {$w_1$};
            \node[above=0mm of 3] {$w_2$};
            \node[above=0mm of 4] {$w_3$};

            \foreach \i in {1, ..., 3} {
              \fill (\i+0.5, -1) circle (2pt) node (u\i) {};
            }

            \path[-{Latex}]
            (u1) edge (2)
            (u1) edge (u2)
            (u1) edge[bend right] (u3)
            (3)  edge (u3)
            (u2) edge (3)
            (u2) edge (u3)
            (2)  edge (u2)
            ;

            \draw[dashed, rounded corners=5,gray!50!black] (1,-1.5) rectangle ++(3, 1) node[below right=0mm] {$F$};
        \end{tikzpicture}
        \caption{Illustration of \cref{lem:short-path-segments}.}
        \label{fig:short-path-segments}
    \end{figure}
\end{proof}

\medskip
\begin{lemma}\label{lem:inner-contribution}
    Every inner path segment contributes to a non-edge or a loop of~$F$.
\end{lemma}
\begin{proof}
Assume towards a contradiction that a non-trivial $P=(w_1,\ldots,w_r)$ does not contribute to a non-edge or a loop of $F$. 
By \cref{lem:short-path-segments} the segment consists only of~$w_1$ and $w_2$. 
Because there is no out-going edge of $w_2$ into $A$, there is no cycle that is induced around $w_1$ that uses $w_2$ (argued exactly as above). 
Hence, $w_1$ is dominated by its unique out-neighbor $z\in A$, and removed and shortcut by \cref{rule:shortcut-new}.
\end{proof}

\begin{lemma}
There are at most $3|F|(|F|-1)k$ inner path segments.
\end{lemma}
\begin{proof}
    As shown in~\cref{lem:inner-contribution}, every inner path segment contributes to a  non-edge or a loop of~$F$. 
    As every inner path segment has a source in $A_2$ as its starting vertex and a sink in $A_2$ as its ending vertex and all inner vertices are balanced vertices, it can intersect with at most two other inner path segments (at their endpoints). In any set~$X$ of inner path segments we can hence find $|X|/3$ independent inner path segments. 
    If there are more than $3|F|(|F|-1)k$ inner path segments that contribute to a  non-edge or a loop, then we find $|F|(|F|-1)k$ many independent ones. Then some pair must be connected by more than $k$ disjoint paths. Then this pair is connected by an edge by \cref{rule:disjoint-cycles} and the path segment does not contribute to it.
\end{proof}

\begin{corollary}\label{cor:num-path-segments}
Overall there are at most $\Oof(|F|^2k)$ path segments. 
\end{corollary}
\medskip

Observe that all path segments are induced paths in $A$. Let $u\in F$ and let $P=(w_1,\ldots, w_r)$ be an induced directed path in $A$ such that $w_1,\ldots, w_{r-1}$ are balanced in~$A$. If $uw_1\in E(G)$ and $uw_r\in E(G)$ and for every $1<i<r$ we have $uw_i\not\in E(G)$, then we call $P$ an \emph{out-segment for $u$}. 
We say that an out-segment for $u$ denoted by $P$ \emph{contributes} to a potential edge or loop $(u,v)$ in $F$ if there is an index $1\leq i<r$ such that $w_iv\in E(G)$ for some $v\in F$. 

\medskip
\begin{lemma}\label{lem:out-segment-contributes}
    Every out-segment for $u$ contributes to a non-edge or loop of~$F$.
\end{lemma}
\begin{proof}
    Assume $P$ is an out-segment for $u$ that does not contribute to a non-edge or loop. Because $w_1$ is balanced in $A$, its only out-neighbor in $A$ is $w_2$, call~$z$ its in-neighbor in~$A$. 
    Hence, every cycle~$C$ using the edge $uw_1$ must contain the vertex $w_r$ or a vertex $v\in F$ with $w_iv\in E(G)$ for some $1\leq i<r$. 
    Note that $v\neq u$ because $P$ does not contribute to a loop. 
    Because~$P$ does not contribute to a non-edge of $F$ in the latter case there is an edge $uv\in E(G)$. 
    This implies that every cycle that is induced around $w_1$ must enter $w_1$ via $z$. 
    Hence, $z$ dominates $w_1$ and $w_1$ is removed and shortcut by \cref{rule:shortcut-new}.
\end{proof}

\begin{lemma}\label{lem:num-out-segments}
    For each $u\in F$ there are at most $|F|k$ out-segments for $u$. 
\end{lemma}

\begin{proof}
According to \cref{lem:out-segment-contributes}, every out-segment for $u$ contributes to a non-edge or loop of $F$. Note that the contribution for $P=(w_1,\ldots, w_r)$ is from the (balanced) vertices $w_1$ and $w_i$, for some $i<r$. Hence, when two out-segments contribute to the same non-edge $(u,v)$ and intersect in $w_r$, the contributing paths still give rise to two internally vertex disjoint $u$-$v$-paths. Hence, there can be at most $k$ out-segments contributing to a non-edge $(u,v)$. Finally, there are at most $|F|$ choices for~$v$.
\end{proof}

We can now bound the size of $A_2$. 

\pagebreak
\medskip
\begin{lemma}\label{lem-final-a2-size}
    $|A_2|\in \Oof(|F|^3k)$.
\end{lemma}
\begin{proof}
    Note that every balanced and source vertex of $A_2$ is an out-neighbor of at least one $u\in F$ and lies on at least one path segment. Fix some $u\in F$. Every out-neighbor of $u$ on a path segment can either be associated to the path segment itself (if $u$ has only one out-neighbor on the whole path segment) or to an out-segment. 
    Therefore, the number of out-neighbors of $u$ in~$A_2$ is at most the number of path segments plus the number of out-segments for $u$. Both numbers are bounded by $\Oof(|F|^2k)$ by \cref{cor:num-path-segments} and \cref{lem:num-out-segments}. Hence, the total number of out-neighbors of vertices of~$F$ in $A_2$, which is equal to the total number of balanced and source vertices in $A_2$, is bounded by $\Oof(|F|^3k)$. Combined with the fact that the number of sink vertices in $A_2$ is also bounded by $\Oof(|F|^2k)$, we get the claimed bound of $|A_2|\in \Oof(|F|^3k)$.
\end{proof}

\begin{theorem}
Let $G$ be a graph and let $f$ denote the size of a minimum feedback vertex set of the underlying undirected graph of $G$. 
After the exhaustive application of \cref{rule:shortcut-new}, \cref{rule:loops}, and \cref{rule:disjoint-cycles} to $G$ we obtain a kernel with  $\Oof(f^3k)$ vertices computable in time $\Oof(n^7(n+m))$.
\end{theorem}

\begin{proof}
The correctness of the kernelization algorithm as well as the size bound follow from the safety of the rules, \cref{cor-013}, and \cref{lem-final-a2-size}. The running time was already analyzed before. \qedhere

\end{proof}
\section{LP-based approximation}\label{sec:lp}

We can derive the following \emph{cycles ILP} for \textsc{DFVS} naturally from the \textsc{Hitting Set} formulation. Given a \textsc{DFVS} instance $G$, we introduce a binary variable $d_v$ for every $v \in V(G)$ where $d_v = 1$ means that $v$ is part of the solution. The goal is to minimize the number of variables set to $1$, given that all induced cycles are hit.
\begin{tcolorbox}
\[\begin{aligned}
\min \sum_{v \in V(G)} d_v \\
\text{s.\,t.} \sum_{v \in V(C)} d_v &\geq 1 \text{ for all induced cycles $C$ in $G$}\\
d_v &\in \{0,1\} \text{ for all } v \in V(G)
\end{aligned}\]
\end{tcolorbox}

Note that this formulation can have an exponential number of constraints. In the following, we assume that our instances contain no loops. We work with the following equivalent \emph{order ILP} of polynomial size which uses the fact that a graph is acyclic if and only if there is a topological order on its vertex set. To be more precise, we order the vertices linearly, minimizing the number of vertices having an incident edge pointing in the incorrect direction. We introduce a binary variable $x_{uv}$ for all distinct $u, v \in V(G)$ where $x_{uv} = 1$ indicates that $u$ is smaller than $v$ in the order. Furthermore, we introduce a binary variable $y_v$ for every $v \in V(G)$ with the same meaning as the~$d_v$ in the cycles ILP.

\begin{tcolorbox}
\[\begin{aligned}
\min \sum_{v \in V(G)} y_v \\
\text{s.\,t. } x_{uv} + x_{vu} &= 1 \text{ for all distinct }u, v \in V(G)\\
x_{uv}+x_{vw}-x_{uw} &\leq 1 \text{ for all distinct }u, v, w \in V(G)\\
x_{uv}+y_u+y_v &\geq 1 \text{ for all } uv \in E(G) \\
x_{uv}, y_v &\in \{0,1\} \text{ for all } u, v \in V(G)
\end{aligned}\]
\end{tcolorbox}

The first two constraints (ensuring anti-symmetry and transitivity) yield a linear order on $V(G)$, and the third constraint ensures that at least one endpoint of any edge pointing in the incorrect direction must be part of the solution.
We prove that the optimal solution value of the relaxation of the order ILP is within a constant factor of the optimal solution value of the relaxation of the all cycles ILP, yielding an approximation benefiting from the practical usability of the simplex algorithm.

\medskip
\begin{theorem}
The optimal solution of the relaxation of the order ILP is at most $3$ times smaller than the optimal solution of the relaxation of the cycles ILP. 
\end{theorem}
\begin{proof}
Fix an optimal solution (assignment of variables) of the relaxation of the order~ILP. Let $C=v_1 \dots v_\ell$ be an induced cycle in $G$. We show that $\sum_{1\leq i\leq \ell}y_{v_i}\geq 1/3$. First assume that all $x_{{v_i}{v_{i+1}}}=1$ for $1\leq i<\ell$. Then by transitivity we have $x_{{v_1}{v_\ell}}=1$ and by anti-symmetry $x_{{v_\ell} {v_1}}=0$. Then by the third constraint $y_{v_1}+y_{v_\ell}\geq 1$ and we are done. 

Hence assume $x_{{v_i}{v_{i+1}}}=1-\epsilon_i$ for some $\epsilon_i\geq 0$, where at least one $\epsilon_i>0$. For every $j\geq i+1$ we have $x_{{v_i}{v_j}}\geq 1-\sum_{i \leq i < q}\epsilon_q$ by transitivity. Hence, $x_{{v_0}{v_\ell}}\geq 1-\sum_{1 \leq q < \ell}\epsilon_q$. 
Then $x_{{v_\ell} {v_1}}\leq \sum_{1 \leq q < \ell}\epsilon_q$ and $y_{v_1}+y_{v_\ell}\geq 1-\sum_{1 \leq q < \ell}\epsilon_q$. 

If $\sum_{1 \leq q < \ell}\epsilon_q\leq 2/3$, then $y_{v_1}+y_{v_\ell}\geq 1/3$ and $C$ collects weight at least~$1/3$. Otherwise we have $\sum_{1 \leq q < \ell}\epsilon_q > 2/3$. Hence, we can write \[\sum_{1\leq i<\ell}x_{{v_i}{v_{i+1}}}=\ell - \sum_{1\leq i<\ell}\epsilon_i\geq \sum_{1\leq i<\ell}(1-y_{v_i}-y_{v_{i+1}})=\ell-y_{v_1}-y_{v_\ell}-2\sum_{1 < i < \ell}y_i.\] Plugging in the inequality we obtain $2/3< \sum_{1\leq i<\ell}\epsilon_i \leq y_{v_1}+y_{v_\ell}+2\sum_{2\leq i<\ell-1}y_{v_i}$. Hence, $C$ collects more than weight $1/3$. 
\end{proof}

Denote the optimal solution value for the relaxation of the cycles ILP by $h^*$, the optimal solution value for the relaxation of the order ILP by $x^*$, and the optimal ILP solution value by $k$. Then 
$k/(4\log 4k \log\ln 4k)\leq h^*\leq 3x^*\leq 3k$,
where the first inequality follows from \cite{seymour1995packing}.

\medskip
\begin{corollary}
We can approximate in polynomial time the relaxation of the cycles ILP up to factor $3$. Hence, we obtain an $\Oof(\log k\log\log k)$ approximation in polynomial time. 
\end{corollary}

\section{Hardness of directed chordless path}\label{sec:hardness}

The \textsc{Directed Chordless $(s,v,t)$-Path} problem asks, given a graph $G$, vertices $s,v,t$, and integer $d$, whether there exists an induced $s$-$t$-path in $G$ of length at most $d$ containing~$v$.  
The W[1]-hardness of the problem on general (directed and undirected) graphs was proved in~\cite{haas2006chordless}. We show hardness on directed acyclic graphs via a reduction from \textsc{Grid Tiling}. 

An instance of \textsc{Grid Tiling} consists of an even integer $k$, an integer $n$, and a collection $\mathcal{S}$ of~$k^2$ nonempty sets $S_{i,j} \subseteq [n] \times [n]$, where $1 \leq i,j \leq k$. The goal is to decide whether there exists, for each  $1 \leq i,j \leq k$, a pair $s_{i.j} \in S_{i,j}$ such that: \\[-3mm]

\begin{itemize}
    \item If $s_{i,j} = (a, b)$ and $s_{i+1,j} = (a',b')$, then $a = a'$.\\[-3mm]
    \item If $s_{i,j} = (a, b)$ and $s_{i,j+1} = (a',b')$, then $b = b'$.\\[-3mm]
\end{itemize}

In other words, if $(i, j)$ and $(i', j')$ are adjacent in the first or second coordinate, then $s_{i,j}$ and $s_{i',j'}$ agree in the first or second coordinate, respectively. We visualize~$S_{i,j}$ to be in a ``cell'' at row $i$ and column $j$ of a ``matrix''. Observe that the constraints
ensure that the first coordinate of the solution is the same in each column and the second coordinate is the same in each row. 

\medskip
\begin{lemma}\label{lem-whard-induced}
The \textsc{Directed Chordless $(s,v,t)$-Path} problem parameterized by the length $d$ of a path is W[1]-hard even when restricted to directed acyclic graphs. 
\end{lemma}

\begin{proof}

\begin{figure}
\centering
\begin{tikzpicture}

    \draw[step=2,black,thin] (0,0) grid (4,4);

    \pic[scale=0.66] at (0.2,3.133) {tile=red!80/cyan!80!black};
    \pic[scale=0.66] at (1.133,2.133) {tile=cyan!80!black/teal!50!black};

    \pic[scale=0.66] at (2.2,3.133) {tile=cyan!80!black/teal!50!black};
    \pic[scale=0.66] at (3.2,2.133) {tile=orange/cyan!80!black};

    \pic[scale=0.66] at (0.5,0.4) {tile=red!80/red!80};

    \pic[scale=0.66] at (2.5,1.2) {tile=teal!50!black/red!80};
    \pic[scale=0.66] at (2.2,0.2) {tile=cyan!80!black/red!80};
    \pic[scale=0.66] at (3.1,0.4) {tile=orange/teal!50!black};

    \node at (5, 2) {$\leadsto$};

    \begin{scope}[xshift=6cm]
        \fill (0.533,3.5) circle (2.5pt) node (v111) {};
        \fill (1.5,2.5) circle (2.5pt) node (v112) {};
        \fill (2.533,3.5) circle (2.5pt) node (v121) {};
        \fill (3.533,2.5) circle (2.5pt) node (v122) {};
        \fill (0.833,0.7333) circle (2.5pt) node (v211) {};
        \fill (2.833,1.5333) circle (2.5pt) node (v221) {};
        \fill (2.533,0.5333) circle (2.5pt) node (v222) {};
        \fill (3.433,0.7333) circle (2.5pt) node (v223) {};

        \draw[dashed, rounded corners=5,gray!50!black] (0.1,0.1) rectangle ++(1.8,1.8);
        \draw[dashed, rounded corners=5,gray!50!black] (2.1,0.1) rectangle ++(1.8,1.8);
        \draw[dashed, rounded corners=5,gray!50!black] (0.1,2.1) rectangle ++(1.8,1.8);
        \draw[dashed, rounded corners=5,gray!50!black] (2.1,2.1) rectangle ++(1.8,1.8);

        \draw[dashed, rounded corners=5,gray!50!black] (0.1,4.1) rectangle ++(3.8, 1);
        \draw[dotted, gray!50!black]                   (0.1,4.6) node[left] {$X$} -- ++(3.8,0);
        \draw[dashed, rounded corners=5,gray!50!black] (0.1,-.1) rectangle ++(3.8,-1);
        \draw[dotted, gray!50!black]                   (0.1,-.6) node[left] {$Y$} -- ++(3.8,0);

        \fill (1, 4.85) circle (2.5pt) node (x1) {} node[left] {$s$};
        \fill (3, 4.35) circle (2.5pt) node (x2) {} node[left] {$t$};
        \fill (1,-.35) circle (2.5pt) node (y1) {} node[left] {$v$};
        \fill (3,-.85) circle (2.5pt) node (y2) {};
        
        \path[-{Latex}]
        (x1) edge (v111)
        (x1) edge (v112)        
        (v211) edge (y1)
        (y1) edge (y2)
        (y2) edge (v221)
        (y2) edge (v222)
        (y2) edge (v223)
        (v121) edge (x2)
        (v122) edge (x2)

        (v111) edge (v211)
        (v222) edge (v121)
        (v223) edge (v122)
        ;

        \path[-{Latex}, red]
        (v111) edge (v121)
        (v112) edge (v122)
        (v211) edge (v223)
        ;
        
    \end{scope}
\end{tikzpicture}
\centering
\caption{Illustration of the construction in the proof of \Cref{lem-whard-induced}.}
\label{fig:example}
\end{figure}

Given an instance of \textsc{Grid Tiling}, we construct an instance of \textsc{Directed Chordless $(s,v,t)$-Path} as follows (see \Cref{fig:example} for an illustration). 
We first construct a directed acyclic graph~$G$. 
For each $S_{i,j} \in \mathcal{S}$, $1 \leq i,j \leq k$, we add a new set of vertices $V_{i,j}$ to $V(G)$. 
$V_{i,j}$ contains one vertex $v^{i,j}_{a,b}$ for each pair $s_{i,j} = (a, b) \in S_{i,j}$. Then we add two new sets of vertices $X = \{x_1, \ldots, x_k\}$ and $Y = \{y_1, \ldots, y_k\}$. We partition $X$ into $X_{odd} = \{x_j \mid j~\text{is odd}\}$ and $X_{even} = \{x_j \mid j~\text{is even}\}$. Similarly, we partition $Y$ into $Y_{odd} = \{y_j \mid j~\text{is odd}\}$ and $Y_{even} = \{y_j \mid j~\text{is even}\}$. 

\pagebreak
We now describe the edges in $G$:\\[-3mm]
\begin{itemize}
    \item For every vertex  $x_j \in X_{odd}$, we add the edges $x_jv^{1,j}_{a,b}$, for all $a,b$. In other words, every vertex in $V_{1,j}$ is made an out-neighbor of $x_j$, $j \in \{1, 3, \ldots, k-1\}$. \\[-3mm]
    \item For every vertex  $x_j \in X_{even}$, we add the edges $v^{1,j}_{a,b}x_j$, for all $a,b$. In other words, every vertex in $V_{1,j}$ is made an in-neighbor of $x_j$, $j \in \{2, 4, \ldots, k\}$. \\[-3mm]
    \item For every vertex  $y_j \in Y_{odd}$, we add the edges $v^{k,j}_{a,b}y_j$, for all $a,b$. In other words, every vertex in $V_{k,j}$ is made an in-neighbor of $y_j$, $j \in \{1, 3, \ldots, k-1\}$. \\[-3mm]
    \item For every vertex  $y_j \in Y_{even}$, we add the edges $y_jv^{k,j}_{a,b}$, for all $a,b$. In other words, every vertex in $V_{k,j}$ is made an out-neighbor of $y_j$, $j \in \{2, 4, \ldots, k\}$.\\[-3mm]
    \item We add the edges $\{y_1y_2, y_3y_4, \ldots, y_{k-1}y_k\} \cup \{x_2x_3, x_4x_5, \ldots, x_{k-2}x_{k-1}\}$.  \\[-3mm]
    \item For odd $j \in \{1, 3, \ldots, k-1\}$ and $i \in [k - 1]$, if there exists $v^{i,j}_{a,b} \in V_{i,j}$ and $v^{i+1,j}_{a',b'} \in V_{i+1,j}$ such that $a = a'$ then add the edge $v^{i,j}_{a,b}v^{i+1,j}_{a',b'}$.\\[-3mm]
    \item For even $j \in \{2, 4, \ldots, k\}$ and $i \in [k] \setminus \{1\}$, if there exists $v^{i,j}_{a,b} \in V_{i,j}$ and $v^{i-1,j}_{a',b'} \in V_{i-1,j}$ such that $a = a'$ then add the edge $v^{i,j}_{a,b}v^{i-1,j}_{a',b'}$. \\[-3mm]
    \item For $j = 1$, $j' \in \{2, 3, 4, \ldots, k\}$,  and $i \in \{1, 2, 3, \ldots, k\}$, if there exists $v^{i,1}_{a,b} \in V_{i,1}$ and $v^{i,j'}_{a',b'} \in V_{i,j'}$ such that $b \neq b'$ then add the edge $v^{i,1}_{a,b}v^{i,j'}_{a',b'}$.\\[-3mm]
\end{itemize}

To complete the construction of the \textsc{Directed Chordless $(s,v,t)$-Path} instance, we choose $s = x_1$, $v = y_1$, $t = x_k$, and $d = k(k + 1) + k - 1 = k^2 + 2k - 1$. 

Observe that $G$ is acyclic since the edges in odd columns are all directed ``downwards'' and the edges in even columns are all directed ``upwards''. Moreover, the edges connecting vertices in $X$ or $Y$ are all directed ``rightwards''. Similarly, the edges connecting vertices in the first column to vertices in later columns are all directed ``rightwards''. Hence, following the layout of the construction, we can topologically order the vertices of $G$ such that for every directed edge $uv$ from vertex $u$ to vertex $v$, $u$ comes before $v$ in the ordering.  

We also note that any $(s,v,t)$-path in $G$ must contain exactly one vertex from each~$V_{i,j}$ as well as all vertices in $X \cup Y$; for a total of $k^2 + 2k$ vertices. Hence, any $(s,v,t)$-path in $G$ will have length exactly $d$. In addition, $P$ must visit $X$, $Y$, and the ``cells'' of the matrix in a unique order. That is, $P$ must start with $x_1$ then visit all the cells of the first column to reach~$y_1$. After~$y_1$ the only out-neighbor is~$y_2$. From $y_2$, $P$ then proceeds upwards along the second column to reach $x_2$. This zig-zag behavior continues until $P$ reaches $x_k$. In other words, $P$ consists of the ordered vertices $s = x_1, \ldots, v = y_1, y_2, \ldots, x_2, x_3, \ldots, y_3, y_4, \ldots, x_4, x_5,$ $\ldots, y_{k-1}, y_{k}, \ldots, t = x_{k}$. 

Assume that we have a yes-instance of \textsc{Directed Chordless $(s,v,t)$-Path} and let $P$ be an induced $(s,v,t)$-path in $G$. The vertices of $V(P) \setminus (X \cup Y)$ correspond one-to-one to pairs in $\mathcal{S}$. We claim that those pairs form a valid solution for the \textsc{Grid Tiling} instance. Assume otherwise. Then, either $P$ contains two consecutive vertices in the same column that do not agree on the first coordinate or $P$ contains two  vertices in the same row that do not agree on the second coordinate. 
The former case is not possible by construction; we only add edges between consecutive vertices in the same column whenever they agree on the first coordinate. For the latter case, we claim that it would contradict the fact that $P$ is induced. To see why, let $i$ be a row containing two vertices $v_2,v_3$ that do not agree on the second coordinate. 
Let~$j_1$ and~$j_2$, $j_1 < j_2$, denote their respective columns. 
If $j_1 = 1$ (that is, $v_2$ belongs to the first column), we have the edge $v_2v_3$ by construction, immediately contradicting the fact that $P$ is induced. Hence assume $1 < j_1 < j_2$.
Recall that~$P$ must include one vertex $v_1$ from~$V_{i.1}$. As we assume that $v_2$ and $v_3$ do not agree on the first coordinate at least one of them does not agree with $v_1$ on the first coordinate.
Hence, by construction, $G$ either contains 
the edge $v_1v_2$ or the edge $v_1v_3$, which contradicts the fact that $P$ is induced. 

Using almost identical arguments, it can be shown that whenever we have a yes-instance of \textsc{Grid Tiling} we can immediately construct an induced $(s,v,t)$-path in $G$ of length exactly~$d$, as needed. 
\end{proof}

By inserting the edge $x_kx_1$, that is, by connecting the top right vertex with the top left vertex with an edge, we obtain the following corollaries of (the proof of) \cref{lem-whard-induced}. 

\medskip
\begin{corollary}
    It is W[1]-hard to decide if a vertex lies on an induced cycle of length at most $d$ even on graphs that become acyclic after the deletion of a single edge.
\end{corollary}

\medskip
\begin{corollary}
    It is W[1]-hard to decide if a graph contains an induced cycle of length at least $d$ even on graphs that become acyclic after the deletion of a single edge.
\end{corollary}
\section{Conclusion}

The question whether \textsc{Directed Feedback Vertex Set} admits a polynomial kernel is one of the most prominent open problems in parameterized complexity. 
First we have presented a new reduction rule subsuming almost all known rules for DFVS, in particular all the complicated rules presented by Bergougnoux~\cite{bergougnoux2021towards}, leading together with two (more or less) trivial rules to a kernel of size $\Oof(k f^3)$, where $f$ denotes the size of a minimum fvs of the underlying undirected graph.
Instances of DFVS without induced cycles of length greater than $d$ correspond to vertex induced instances of $d$-\textsc{Hitting Set}, for which it is an important open problem whether it admits a kernel with fewer than $\Oof(k^{d-\epsilon})$ elements. 
On the one hand, since it is a special case, it is perceivable that DFVS without induced cycles of length greater than~$d$ is easier to kernelize than the more general $d$-\textsc{Hitting Set}. 
On the other hand, it is not clear that all data reduction rules for $d$-\textsc{Hitting Set} can be carried out efficiently for DFVS instances. 
The caveat is that \textsc{Hitting Set} instances are explicitly represented, whereas the \textsc{Hitting Set} instance corresponding to DFVS is only implicitly represented in the input graph. 
As a consequence we do not allow a running time that depends polynomially on the number of cycles, but only polynomially on the number of vertices of the graph. 

In this work we proved that DFVS without cycles of length greater than $d$ admits a kernel with $2^dk^d$ vertices. 
For this we designed a new reduction rule that approximates the set of vertices that lie on induced cycles. 
We proved that even in acyclic graphs we cannot efficiently test whether a vertex lies on an induced path between two given vertices, \textsc{Directed Chordless} $(s,v,t)$-\textsc{Path} is W[1]-hard when parameterized by $d$. 
This shows that a very natural reduction rule is not efficiently applicable for DFVS. 
We then considered restricted graph classes and proved that in very general classes kernels with an almost linear number $\Oof(k^{1+\epsilon})$ of vertices computable in time $f(d)\cdot n^{\Oof(1)}$ exist. 
We repeat the question whether this is possible in general graphs. 

\bibliography{ref}

\end{document}